\documentclass[11pt]{article}

\usepackage[utf8]{inputenc}
\usepackage[colorlinks=true
]{hyperref}
\hypersetup{urlcolor=blue, citecolor=red, linkcolor=blue}
\usepackage{mleftright}
\usepackage{stmaryrd}
\SetSymbolFont{stmry}{bold}{U}{stmry}{m}{n}
\usepackage{microtype}

\usepackage{bbm}
\usepackage{colortbl}

\usepackage{graphicx}
\usepackage[active]{srcltx}

\usepackage{amsmath,amssymb}

\numberwithin{equation}{section}

\usepackage{physics}
\usepackage{bm}
\usepackage{enumerate}
\usepackage{amsthm}
\usepackage[all,2cell]{xy}
\usepackage{mathrsfs}
\usepackage{mathtools}
\mathtoolsset{showonlyrefs}

\usepackage{tikz-cd}
\usetikzlibrary{cd}

\usepackage{xcolor}
\DeclareUnicodeCharacter{0301}{*************************************}

\hypersetup{urlcolor=blue, citecolor=red, linkcolor=blue}

\usepackage[a4paper, left=25mm, right=25mm, top=30mm, bottom=25mm]{geometry}

\usepackage[symbol]{footmisc}

\usepackage{marginnote}

\usepackage{imakeidx}
\makeindex[intoc]

\usepackage[nottoc]{tocbibind}
\usepackage[colorinlistoftodos]{todonotes} 

\theoremstyle{plain}
\newtheorem{theorem}{Theorem}[section]
\newtheorem{proposition}[theorem]{Proposition}
\newtheorem{corollary}[theorem]{Corollary}

\theoremstyle{definition}
\newtheorem{definition}[theorem]{Definition}
\newtheorem{remark}[theorem]{Remark}
\newtheorem{example}[theorem]{Example}

\AtBeginEnvironment{remark}{%
  \pushQED{\qed}%
}
\AtEndEnvironment{remark}{\popQED\endexample}

\AtBeginEnvironment{example}{%
  \pushQED{\qed}%
}
\AtEndEnvironment{example}{\popQED\endexample}

\newcommand\restr[2]{{
  \left.\kern-\nulldelimiterspace #1 \right|_{#2} 
}}

\renewcommand{\d}{\mathrm{d}}

\newcommand{\df}{\Omega}
\newcommand{\Cinfty}{\mathscr{C}^\infty}

\newcommand{\Tan}{\mathrm{T}}
\newcommand{\cT}{\mathrm{T}^\ast}

\newcommand*{\inn}[1]{\iota_{#1}}
\newcommand*{\innp}[1]{\iota\left(#1\right)}

\newcommand{\Lie}{\mathscr{L}}
\newcommand{\vf}{\mathfrak{X}}

\newcommand{\W}{\mathcal{W}}

\renewcommand{\W}{\mathcal{W}}
\newcommand*{\bd}{\overline{\mathrm{d}}}

\newcommand{\bfX}{\mathbf{X}}

\def\derpar#1#2{\frac{\partial{#1}}{\partial{#2}}}
\newcommand{\parder}[2]{\frac{\partial #1}{\partial #2}}

\newcommand{\tparder}[2]{\partial #1/\partial #2}
\newcommand{\parderr}[3]{\frac{\partial^2 #1}{\partial #2\partial #3}}

\let\graph\relax
\DeclareMathOperator{\graph}{graph}

\DeclareMathAlphabet{\mathpzc}{OT1}{pzc}{m}{it}
\def\d{\mathrm{d}}

\DeclareMathOperator{\pr}{pr}

\usepackage{xcolor}
\usepackage{fancyhdr}

\begin{document}


\vspace{3em}

{\Huge\sffamily\raggedright Skinner--Rusk formalism \\[12pt] of action-dependent multicontact field theories
}
\vspace{2em}

{\large\raggedright
    \today
}

\vspace{2em}

\medskip

{\Large\raggedright\sffamily
    Xavier Rivas
}\vspace{1mm}\newline
{\raggedright
    Department of Computer Engineering and Mathematics, Universitat Rovira i Virgili.\\
    Avinguda Països Catalans 26, 43007 Tarragona, Spain.\\
    e-mail: \href{mailto:xavier.rivas@urv.cat}{xavier.rivas@urv.cat} --- {\sc orcid}: \href{https://orcid.org/0000-0002-4175-5157}{0000-0002-4175-5157}
}

\medskip

{\Large\raggedright\sffamily
    Narciso Román-Roy
}\vspace{1mm}\newline
{\raggedright
    Department of Mathematics, Universitat Politècnica de Catalunya.\\
    C. de Jordi Girona 31, 08034 Barcelona, Spain.\\
    e-mail: \href{mailto:narciso.roman@upc.edu}{narciso.roman@upc.edu } --- {\sc orcid}: \href{https://orcid.org/0000-0003-3663-9861}{0000-0003-3663-9861}
}

\medskip

{\Large\raggedright\sffamily
    Annamaria Villanova
}\vspace{1mm}\newline
{\raggedright
    Escuela Superior de Ingeniería y Tecnología, Universidad Internacional de La Rioja.\\
    Av. de la Paz, 137, 26006 Logroño, La Rioja, Spain.\\
    Department of Mathematics, Universitat Politècnica de Catalunya.\\
    C. Jordi Girona 31, 08034 Barcelona, Spain.\\
    e-mail: \href{mailto:annamaria.villanova@unir.net}{annamaria.villanova@unir.net} --- {\sc orcid}: \href{https://orcid.org/0009-0002-2350-8109}{0009-0002-2350-8109}
    \medskip
}

\vspace{2em}

{\large\bf\raggedright
    Abstract
}\vspace{1mm}\newline
{\raggedright
The newly developed {\sl multicontact structure},
based on {\sl contact} and {\sl multisymplectic geometries}, provides a very general geometrical framework suitable for the treatment of action-dependent classical field theories.
Having successfully applied it to formulate the Lagrangian and Hamiltonian descriptions of these theories, in the present work, the well-known Skinner--Rusk formalism is presented in this multicontact setting, which allows us to provide a combined version of both Lagrangian and Hamiltonian formalisms particularly suitable for the study and description of singular theories.
As an application of this last situation, we study a modification of Maxwell’s Lagrangian of classical electromagnetism, which incorporates action-dependent terms and allows us to describe electromagnetism in material media.
}

\vspace{1em}


\medskip

{\large\bf\raggedright
MSC2020 codes}:
{\sl Primary:}
70S05, 
70S10,
53D10,
35R01. \\
\indent\indent\indent\indent\indent
{\sl Secondary:}
35Q99, 53C15, 53Z05, 58A10, 58J60, 70G45.

\vspace{2em}

\noindent {\bf Authors' contributions:} All authors contributed to the study conception and design. The manuscript was written and revised by all authors. All authors read and approved the final version.
\medskip

\noindent {\bf Competing Interests:} The authors have no competing interests to declare. 

\newpage

{\setcounter{tocdepth}{2}
\def\baselinestretch{1}
\small
\def\addvspace#1{\vskip 1pt}
\parskip 0pt plus 0.1mm
\tableofcontents
}

\pagestyle{fancy}

\fancyhead[L]{Skinner--Rusk formalism of multicontact field theories}    
\fancyhead[C]{}                  
\fancyhead[R]{X. Rivas {\it et al.}}       

\fancyfoot[L]{}     
\fancyfoot[C]{\thepage}                  
\fancyfoot[R]{}            

\setlength{\headheight}{17pt}

\renewcommand{\headrulewidth}{0.1pt}  
\renewcommand{\footrulewidth}{0pt}    

\renewcommand{\headrule}{%
    \vspace{3pt}                
    \hrule width\headwidth height 0.4pt 
    \vspace{0pt}                
}

\setlength{\headsep}{30pt}  

\section{Introduction}

In the last decades, the description of physical theories has experienced an interesting development through the use of differential geometric methods. This has proven to be a very powerful and insightful approach, providing a universal language to describe and analyze physical systems. In particular, for classical field theories, \textsl{multisymplectic geometry} has been established as the most standard and general framework for describing conservative systems \cite{Gar_74,GS_73,Kij_73,KT_79}.
Although M.J. Gotay was the first to give an abstract definition of {\sl multisymplectic form} \cite{Go_91},
this geometric structure was later studied in detail by
other authors \cite{CIL_96,CIL_99,CCI_91,LMS_03,EIMR_12,Mar_88,RW_19}.

A variation of classical field theories are \textsl{action-dependent field theories}, namely those providing
non-conservation instead of the usual conservation laws.
The Lagrangian or Hamiltonian functions that describe these theories are allowed to depend not only on the fields and their derivatives, but also on additional variables that encode the flow of action.
The key and more general geometric structure that arises in the formulation of these theories is \textsl{multicontact geometry} \cite{LGMRR_23,LIR_25,Vit_15}, a generalization of contact geometry \cite{BH_16,Bra_17,CG_19,Gei_08,Kho_13,LM_87}
based on the multisymplectic setting,
which constitutes the most general geometric framework for dealing with action-dependent field theories.
Action-dependent field theories can also be derived and studied within the framework of variational calculus \cite{LGMRR_23,GLMR_24,GM_22,GGB_03,LPAF_2018}, which provides more tools to model dissipative systems.
One of the distinctive features of action-dependent theories is that their dissipation laws are precisely governed by geometric data derived from the Lagrangian or Hamiltonian; 
that is, through a differential 1-form known as the \textsl{dissipation form}.
The ability to quantify these effects directly from the geometry is a major advantage of the multicontact approach.

In 1983, R. Skinner and R. Rusk developed a geometric formalism by merging the Lagrangian and Hamiltonian formalisms of autonomous dynamical systems into a single unified description,
where velocities and momenta are treated within a unified framework \cite{SR_83,SR_83_2}.
Regardless of the regularity of the Lagrangian that describes the system, in the Skinner--Rusk formalism the theory is always singular.
A key property of this formalism is that the dynamical equations inherently incorporate the second-order or holonomy condition
for the solutions, regardless of whether the Lagrangian is regular or not,
as well as the Legendre map, which is obtained as compatibility conditions for the dynamical equations.
This unified framework is especially powerful in the context of singular theories, and the Lagrangian and Hamiltonian versions of the constraint algorithm, as well as the corresponding
solutions to the Euler--Lagrange and Hamilton equations, are recovered straightforwardly.

This Skinner--Rusk formalism was later extended to other kinds of mechanical systems and classical field theories, and, recently, also in the context
of contact mechanics \cite{LGLRR_20,LGMR_21,RT_23} and $k$-contact field theories \cite{GRR_22} (see also the references therein for an exhaustive bibliography on the generalizations of the Skinner--Rusk formalism).

In this paper, we extend the Skinner--Rusk formalism to the multicontact framework, in the context of action-dependent field theories,
developing a unified structure that integrates the two classical Lagrangian and Hamiltonian formalisms, which are recovered as specific cases of the Skinner--Rusk formalism. We study the relevant geometric structures which arise in this context and use a modification of Maxwell's electromagnetism to illustrate the theory.
This unified approach is shown to be more comprehensive, since the resulting field equations contain all the relevant features of both the Lagrangian and Hamiltonian descriptions of the theory.

The paper is organized as follows. Section \ref{reviw} provides a review of the general structure of multicontact manifolds and of the Lagrangian and Hamiltonian formalisms for action-dependent field theories.
Section \ref{skinner--Rusk} contains the main original contributions of this work. In particular, it is devoted to studying the geometric frameworks mentioned above. To give the mathematical theory a physical relevance, Section \ref{example} proposes, as an application, the well-known theory of Maxwell's electromagnetism modified by incorporating extra action-dependent terms.

All manifolds are real, second-countable, and smooth, and the mappings are assumed to be of class $\mathscr{C}^\infty$.
Einstein summation over repeated indices is understood.
The following notation will be used throughout the work, adhering to standard conventions:
\vspace{-10pt}

\begin{itemize}\itemsep1pt
    \item $\Cinfty({\cal M})$: Smooth functions in a manifold ${\cal M}$.
    \item $\df^k({\cal M})$: Module of differential forms of degree $k$ in a manifold ${\cal M}$.
    \item $\vf({\cal M})$: Module of vector fields in a manifold ${\cal M}$.
    \item $\vf^k({\cal M})$: Module of $k$-multivector fields in a manifold ${\cal M}$.   
    \item $\inn{X}\Omega$: Inner contraction of a vector field $X\in\vf({\cal M})$ and a $k$-form $\Omega\in\df^k({\cal M})$.
    \item $\Lie_X$: Lie derivative by a vector field $X\in\vf({\cal M})$.
    \item $\d$: Exterior differential of differential forms.
    \end{itemize}

\section{Multicontact formulation of action-dependent field theories} 
\label{reviw}

The {\sl multicontact formulation} of action-dependent classical field theories
is a very general geometric framework recently developed to describe these kinds of theory.
It is based on an extension of the contact and multisymplectic geometries, called {\sl multicontact structure}, which was first introduced in \cite{LGMRR_23,LGMRR_25},
although a more general definition of this structure has recently been proposed in \cite{LIR_25}.
The next review summarizes the main concepts and results of these references and follows
the terminology introduced in \cite{GRR_25}.

\subsection{Multivector fields and connections}
\label{ap:multivector}

This section reviews the basics on multivector fields and connections, including the integrability conditions of these objects. For more details, see \cite{BCGGG_91,CCI_91,EIMR_12,EMR_98,KMS_93,Mar_97}.

Let ${\cal M}$ be a manifold with $\dim{\cal M}=n$.
The {\sl\textbf{$m$-multivector fields}} or {\sl\textbf{multivector fields of grade $m$}} on ${\cal M}$ are the contravariant skew-symmetric tensor fields of order $m$ in ${\cal M}$. 
The set of $m$ multivector fields in ${\cal M}$ is denoted by $\vf^m({\cal M})$.
A multivector field $\bfX\in \vf^m({\cal M})$ is {\sl\textbf{locally decomposable}} if, for every ${\rm p}\in{\cal M}$, there exists an open neighbourhood $U_{\rm p} \subset{\cal M}$ such that 
\begin{equation}
\restr{\bfX}{U_{\rm p}} = X_1\wedge\dotsb\wedge X_m \,,\quad \text{for some }X_1,\dotsc,X_m \in\vf(U_{\rm p}) \,.
\end{equation}
The {\sl \textbf{contraction}} of a locally decomposable multivector field $\bfX\in \vf^m({\cal M})$ and a differentiable form $\Omega\in\df^k({\cal M})$
is the natural contraction between contravariant and covariant tensor fields,
\begin{equation}
    \restr{\inn{\bfX}\Omega}{U_p} = \begin{dcases}
        \inn{X_1\wedge\dotsb \wedge X_m}\Omega = \inn{X_m}\dotsb\inn{X_1}\Omega\,, & k \geq m\,,\\
        0\,, & k < m\,.
    \end{dcases}
\end{equation}

Let $\kappa\colon {\cal M}\rightarrow M$ be a fiber bundle
with $\dim{M}=m$ and $\dim{{\cal M}}=N+m$,
and denote $(x^\mu,z^A)$ the local coordinates on ${\cal M}$;
where $x^\mu$ are coordinates on the base $M$ and $z^A$ are coordinates on the fibers (with $0\leq\mu\leq {m-1}$ and $1\leq A\leq N$).
A multivector field $\bfX \in \vf^m({\cal M})$  is $\kappa$-\textsl{transverse} if
$\inn{\bfX}(\kappa^*\beta)\vert_{\rm p} \neq 0$,
for every ${\rm p} \in{\cal M}$ and $\beta \in \df^m({M})$.
The local expression for a locally decomposable $\kappa$-transverse multivector field $\bfX\in\vf^m({\cal M})$ is 
\begin{equation}
\bfX= \bigwedge^{m-1}_{\mu=0} \bfX_\mu= 
f\bigwedge^{m-1}_{\mu=0} \left( \frac{\partial}{\partial x^\mu} + F_\mu^A\frac{\partial}{\partial z^A}\right)\,, \quad \text{where }f\in \Cinfty({\cal M}) \,.
\label{mvf}
\end{equation}
 If $M$ is an orientable manifold with volume form $\omega_M\in\df^m(M)$, then
the condition for $\bfX \in \mathfrak{X}^m({\cal M})$ to be $\kappa$-transverse can be expressed as
$\inn{\bfX}(\kappa^*\omega_M)\neq 0$.
This transversality condition can be fixed by taking
$\inn{\bfX}(\kappa^*\omega_M)=1$,
which implies $f=1$ in \eqref{mvf}.

 If $\bfX \in \vf^m({\cal M})$ is a locally decomposable and $\kappa$-transverse multivector field
and $\psi\colon U\subset M \rightarrow{\cal M}$
is a local section of $\kappa$,
with local expression
$\psi(x^\mu)=(x^\mu,z^A(x^\nu))$;
then $\psi$ is an {\sl\textbf{integral section}} of $\bfX$ if $\displaystyle\derpar{z^A}{x^\mu}=F_\mu^A$.
Then, $\bfX$ is {\sl \textbf{integrable}} if,
for every ${\rm p}\in{\cal M}$, there exist $x\in M$ and an integral section, $\psi\colon U_x\subset M\to {\cal M}$, of $\bfX$ such that ${\rm p} = \psi(x)$.

Furthermore, the {\sl\textbf{canonical prolongation}} of a section 
$\psi\colon U\subset M\to \mathcal{M}$ to 
$\bigwedge\nolimits^m\Tan{\cal M}$ is
the section $\psi^{(m)}\colon U\subset M\to\bigwedge\nolimits^m\Tan\mathcal{M}$ 
defined as $\psi^{(m)}:=\bigwedge\nolimits^m\Tan\psi\circ{\bf Y}_{\omega_M}$, where 
$\bigwedge\nolimits^m\Tan\psi\colon\bigwedge\nolimits^m\Tan M\to\bigwedge\nolimits^m\Tan\mathcal{M}$ is the natural extension of $\psi$ 
to the corresponding multitangent bundles,
and ${\bf Y}_{\omega_M}\in\vf^m(M)$ is the unique $m$-multivector field on $M$
such that $\inn{{\bf Y}_{\omega_M}}\omega_M=1$. Then,
$\psi$ is an integral section of ${\bf X}\in\vf^m({\cal M})$ if, and only if, ${\bf X}\circ\psi=\psi^{(m)}$.

An {\sl\textbf{Ehresmann connection}} on the bundle $\mathcal{M}\to M$
is a $\kappa$-semibasic $1$-form $\nabla$ on $\mathcal{M}$
with values in $\Tan \mathcal{M}$, that is, a $(1,1)$-tensor field on $\mathcal{M}$,
$\nabla\in\mathcal T^1_1(\mathcal{M})$,
such that $\inn{\nabla}\alpha =\alpha$, for every
$\kappa$-semibasic form $\alpha\in\df^1(\mathcal{M})$
(where $\inn{\nabla}\alpha$ denotes the usual tensor contraction).
An Ehresmann connection splits $\Tan \mathcal{M}$ into the {\sl vertical} and a {\sl horizontal distribution} and, in this way, $\nabla$ represents also the horizontal projector.
The connection is said to be {\sl\textbf{integrable}} if its associated horizontal distribution is integrable
(the necessary and sufficient condition is that 
the connection is {\sl flat}).
In this case, the integral submanifolds are local sections $\psi\colon U\subset M\to \mathcal{M}$ of $\kappa$.

Then, classes of locally decomposable and
$\kappa$-transverse multivector fields $\{{\bf X}\} \subseteq \vf^m (\mathcal{M})$
are in one-to-one correspondence with orientable Ehresmann connections $\nabla$ on $\mathcal{M}\to M$. This correspondence is
characterized by the fact that the horizontal distribution associated
with $\nabla$ is the distribution associated with $\{ {\bf X}\}$.
In this correspondence,
classes of integrable (locally decomposable) and $\kappa$-transverse
$m$-multivector fields correspond to flat orientable Ehresmann
connections.
In natural coordinates of the bundle, the local expression of a connection is:
$$
\nabla= \d x^\mu\otimes\left( \frac{\partial}{\partial x^\mu} + F_\mu^A\frac{\partial}{\partial z^A}\right) \,.
$$

\subsection{Multicontact structures}

\begin{definition}
\label{multicont}
Let $\mathcal{M}$ be a differentiable manifold.
A form $\Theta\in\df^m(\mathcal{M})$, with $\dim\mathcal{M}>m$,
is a {\sl \textbf{multicontact form}} in $\mathcal{M}$ if:
\begin{enumerate}[{\rm(1)}]
    \item
$\ker \Theta\cap\ker\d \Theta=\{0\}$, and
    \item 
$\ker\d\Theta\neq\{0\}$.
    \end{enumerate}
    Then, the pair $(\mathcal{M},\Theta)$ is said to be a {\sl \textbf{multicontact manifold}}.
If condition (1) does not hold, then $\Theta$ 
is said to be a {\sl \textbf{premulticontact form}}
and $(\mathcal{M},\Theta)$ is a {\sl \textbf{premulticontact manifold}}.
\end{definition}

This definition is very general and,
as in the case with multisymplectic structures \cite{LMS_03,Mar_88}, 
the existence of adapted or Darboux coordinates
as well as other useful properties is not guaranteed for these multicontact forms, 
unless additional conditions are imposed.
Thus, let $\mathcal{M}$ be a manifold,  with $\dim{\mathcal{M}}=m+N$ and $N\geq m\geq 1$;
and let $\Theta,\omega\in\df^m(\mathcal{M})$ be two forms with constant ranks.
Given a regular distribution $\mathcal{D}\subset\Tan \mathcal{M}$, consider the $\Cinfty(\mathcal{M})$-module of sections $\Gamma(\mathcal{D})$ and, for every $k\in\mathbb{N}$, define
the set of $k$-forms on $\mathcal{M}$ vanishing by the vector fields in $\Gamma(\mathcal{D})$; that is,
\[
\mathcal{A}^k(\mathcal{D}):=\big\{ \alpha\in\df^k(\mathcal{M})\mid
\inn{Z}\alpha=0 \,,\ \text{ for every}\ Z\in\Gamma({\cal D})\big\}=
\big\{ \alpha\in\df^k(\mathcal{M}) \mid 
\Gamma(\cal D)\subset\ker\alpha \big\} \,,
\]
where $\ker\alpha =\{ Z\in\vf(\mathcal{M})\mid \inn{Z}\alpha=0\}$ is the one-kernel of a $k$-form $\alpha\in\df^k(\mathcal{M})$, with $k>1$.

\begin{definition}
\label{def:reeb_dist}
The {\sl\textbf{Reeb distribution}} associated with
the pair $(\Theta,\omega)$ is the distribution
$\mathfrak{R}\subset\Tan \mathcal{M}$ defined as
\begin{equation}
\label{Reebdef}
\mathfrak{R}=\big\{ R\in\Gamma(\ker\omega)\mid \inn{R}\dd\Theta\in {\cal A}^m(\ker\omega)\big\} \,.
\end{equation}
The set of sections of the Reeb distribution will be also denoted by $\mathfrak{R}$,
and its elements $R\in\mathfrak{R}$ are called {\sl\textbf{Reeb vector fields}}.
\end{definition}

Note that $\ker\omega\cap\ker\dd\Theta\subset\mathfrak{R}$.
Furthermore, if $\omega\in\df^k(\mathcal{M})$ is a closed form and has a constant rank, then $\mathfrak{R}$ is involutive. Therefore:

\begin{definition}
\label{multicontactdef}
A pair $(\Theta,\omega)$ is a {\sl \textbf{special premulticontact structure}} on $\mathcal{M}$ if $\omega\in\Omega^m(\mathcal{M})$ is closed and, for $0\leq k\leq N-m$,
we have the following:
\begin{enumerate}[{\rm (1)}]
\item\label{prekeromega}
$\rank\ker\omega=N$,
\item\label{prerankReeb}
$\rank\mathfrak{R}=m+k$,
\item\label{prerankcar}
$\rank\left(\ker\omega\cap\ker\Theta\cap\ker\d\Theta\right)=k$,
\item \label{preReebComp}
${\cal A}^{m-1}(\ker\omega)=\{\inn{R}\Theta\mid R\in \mathfrak{R}\}$.
\end{enumerate}
Then, the triple $(\mathcal{M},\Theta,\omega)$ is said to be
a {\sl \textbf{special premulticontact manifold}}
and $\Theta\in\Omega^m(\mathcal{M})$ is called a {\sl \textbf{special premulticontact form}} on $\mathcal{M}$. 
The distribution
$\mathcal{C} = \ker\omega\cap\ker\Theta\cap\ker\d\Theta$
is the {\sl \textbf{characteristic distribution}} of $(\mathcal{M},\Theta,\omega)$.
If $k=0$, the pair $(\Theta,\omega)$ is a {\sl \textbf{special multicontact structure}}, the triple
$(\mathcal{M},\Theta,\omega)$ is a {\sl \textbf{special multicontact manifold}},
and $\Theta\in\Omega^m(\mathcal{M})$ is a {\sl \textbf{special multicontact form}}.
\end{definition}

\begin{remark}
As it is expected, if $(\mathcal{M},\Theta,\omega)$ is a special multicontact manifold, then $(\mathcal{M},\Theta)$ is a multicontact manifold \cite{GRR_25}.
\end{remark}

\begin{remark}
We will be interested in the particular case  of special (pre)multicontact manifolds $(\mathcal{M},\Theta,\omega)$,
where $\kappa\colon\mathcal{M}\to M$ is a fiber
bundle, with $\dim\, M= m$, $\dim\,\mathcal{M}=m+N$, 
and such that $M$ is an orientable manifold with volume form 
$\omega_M\in\df^m(M)$ and $\omega=\kappa^*\omega_M\in\df^m(\mathcal{M})$.
Then condition (1) of Definition\ref{multicontactdef} holds automatically and we say that $(\mathcal{M},\Theta,\omega)$ is a {\sl\textbf{special (pre)multicontact bundle}}.
This is, really, the canonical model for special (pre)multicontact structures
and is also the situation which is interesting in action-dependent field theories.
\end{remark}

\begin{proposition}
Let $(\mathcal{M},\Theta,\omega)$ be a special (pre)multicontact manifold, then
there exists a unique $1$-form
$\sigma_{\Theta}\in\df^1(\mathcal{M})$, called the {\sl\textbf{dissipation form}}, satisfying
\begin{equation}
\label{sigma}
\sigma_{\Theta}\wedge\inn{R}\Theta=\inn{R}\dd\Theta \,,\qquad \text{for every }\,R\in\mathfrak{R} \,.
\end{equation}
\end{proposition}

\begin{definition}\label{def:bard}
\label{bard}
Let $\sigma_{\Theta}\in\df^1(\mathcal{M})$ be the dissipation form. 
We define the operator
\begin{align*}
\bd\colon\df^m(\mathcal{M})&\longrightarrow\df^{m+1}(\mathcal{M})
\\
\beta&\longmapsto\bd\beta=\dd \beta+\sigma_{\Theta}\wedge\beta\,.
\end{align*}
\end{definition}

The (pre)multicontact structures corresponding to action-dependent field theories arising from the Herglotz variational principle (see \cite{GLMR_24}) satisfy the following additional requirement.

\begin{definition}\label{def:variational}
Let \((\mathcal{M},\Theta,\omega)\) be a special (pre)multicontact manifold satisfying
\begin{equation}
\iota_X\iota_{Y}\Theta = 0 \,, \qquad
\text{for every } X,Y\in\Gamma(\ker\omega) \,.
\end{equation}
Then $(\mathcal{M},\Theta,\omega)$ is a {\sl\textbf{variational (pre)multicontact manifold}} and $(\Theta,\omega)$ is said to be a {\sl\textbf{variational (pre)multicontact structure}}.
\end{definition}

\subsection{Multicontact Lagrangian formalism}
\label{mlf}

Consider a bundle $\pi\colon E\to M$,
where $M$ is an orientable $m$-dimensional manifold with volume form $\omega_M\in\df^m(M)$, and let $J^1\pi\to E\to M$ be the corresponding first-order jet bundle.
If $\dim M=m$ and $\dim E=n+m$, then $\dim J^1\pi=nm+n+m$.
The natural coordinates in $J^1\pi$ adapted to the bundle structure
are $(x^\mu,y^A,y^A_\mu)$, with $\mu = 1,\ldots,m$ and $A=1,\ldots,n$,
and are taken so that
$\omega_M=\d x^1\wedge\cdots\wedge\d x^m =: \d^mx$.

In the multicontact Lagrangian formalism for action-dependent field theories,
the {\sl configuration bundle} of the theory is $E\times_M\bigwedge\nolimits^{m-1}(\Tan^*M)\to M$,
where $\bigwedge\nolimits^{m-1}(\Tan^*M)$ denotes the bundle of $(m-1)$-forms on $M$.
The  {\sl multivelocity phase bundle} is
${\cal P}=J^1\pi\times_M\bigwedge\nolimits^{m-1}(\Tan^*M)$.
Natural coordinates in ${\cal P}$ are $(x^\mu,y^A,y^A_\mu,s^\mu)$,
and $\dim\mathcal{P}=2m+n+nm$.
We have the natural projections depicted in the following diagram:
\begin{equation}\label{diagrammulticoLag}
\begin{tikzcd}[column sep=huge, row sep=large]
    & {\cal P}=J^1\pi\times_M\bigwedge\nolimits^{m-1}\Tan^*M
    \arrow[dl, "\rho", swap] 
    \arrow[ddd, "\tau"] 
    \arrow[dr, "\tau_1"] 
    \arrow[ddl, "\rho_E"] \\
    J^1\pi
    \arrow[d, "\pi^1", swap]
    \arrow[ddr, "\bar\pi^1", crossing over, yshift = 1.0]
    & & \bigwedge\nolimits^{m-1}\Tan^*M
    \arrow[ddl, "\tau_\circ"] \\
    E
    \arrow[dr, "\pi", swap, end anchor={[yshift=1.5ex]south west}] \\
    & M
\end{tikzcd}
\end{equation}

As $\bigwedge\nolimits^{m-1}\Tan^*M$ is a bundle of forms, it is endowed with a canonical structure, the {\sl tautological form}
$\theta\in \df^{m-1}(\bigwedge\nolimits^{m-1}\Tan^*M)$,
whose expression in natural coordinates is $\theta=s^\mu\,\d^{m-1}x_\mu$. Then:

\begin{definition}
The form $S:=\tau_1^*\theta\in\df^{m-1}({\cal P})$
is called the {\sl\textbf{canonical action form}} of ${\cal P}$.
\end{definition}

Its expression in coordinates is also $S=s^\mu\,\d^{m-1} x_\mu$.

A section $\bm{\psi}\colon M\rightarrow {\cal P}$ of the projection $\tau:\mathcal{P}\rightarrow M$ is said to be {\sl\textbf{holonomic}} in ${\cal P}$ if
the section $\psi:=\rho\circ\bm{\psi}\colon M\to J^1\pi$
is holonomic in $J^1\pi$; that is,
there is a section $\phi\colon M\to E$ of $\pi$ such that $\psi=j^1\phi$.
It is customary to write $\bm{\psi}=(\psi,s)=(j^1\phi,s)=j^1\bm\phi$, 
where $s\colon M\to\bigwedge\nolimits^{m-1}(\Tan^*M)$  is a section of the projection $\tau_\circ\colon\bigwedge\nolimits^{m-1}(\Tan^*M)\to M$;
then, we also say that $\bm{\psi}$ is the
{\sl\textbf{canonical prolongation}} of the section 
$\bm\phi:=(\phi,s)\colon M\to E\times_M\bigwedge\nolimits^{m-1}(\Tan^*M)$ to ${\cal P}$.

Now, 

a $\tau$-transverse $m$-multivector field
$\bm{\Gamma}\in\vf^m({\cal P})$ or an Ehresmann connection $\nabla\in\mathcal T^1_1(\mathcal{M})$ are said to be {\sl\textbf{holonomic}}
or a {\sl\textbf{second-order partial differential equation}} (\textsc{sopde}) in ${\cal P}$ if
they are integrable and its integral sections are holonomic on ${\cal P}$.

The local expression of a {\sc sopde} multivector field in ${\cal P}$, satisfying the condition $\inn{\bfX}\omega=1$, is 
\begin{equation}
    \label{localsode2}
    \bfX = \bigwedge\nolimits^m_{\mu=1}
    \Big(\parder{}{x^\mu}+y^A_\mu\frac{\displaystyle\partial} {\displaystyle
    \partial y^A}+F_{\mu\nu}^A\frac{\displaystyle\partial}{\displaystyle \partial y^A_\nu}+g^\nu_\mu\,\frac{\partial}{\partial s^\nu}\Big)\,,
\end{equation}
and, similarly, the local expression of a {\sc sopde} connection is
\begin{equation}
    \label{localsode22}
    \nabla=\d x^\mu\otimes
    \Big(\derpar{}{x^\mu}+y^A_\mu\frac{\displaystyle\partial} {\displaystyle
    \partial y^A}+F_{\mu\nu}^A\frac{\displaystyle\partial}{\displaystyle \partial y^A_\nu}+g^\nu_\mu\,\frac{\partial}{\partial s^\nu}\Big)\,.
\end{equation}
Their integral sections are solutions to the system of second-order partial differential equations:
\begin{equation}
    y^A_\mu=\parder{y^A}{x^\mu}\,,\qquad F^A_{\mu\nu}=\frac{\partial^2y^A}{\partial x^\mu \partial x^\nu}\,.
\end{equation}

The $\tau$-transverse locally decomposable $m$-multivector fields and Ehresmann connections in~${\cal P}$ 
whose local expressions are \eqref{localsode2}
and \eqref{localsode22}, respectively,
are usually referred to as {\sl\textbf{semi-holonomic multivector fields}} and {\sl\textbf{semi-holonomic connections}}.

The first-order jet bundle $J^1\pi$ is endowed with a canonical structure
which is called the {\sl canonical endomorphism},
and is a $(1,2)$-tensor field in $J^1\pi$,
denoted ${\rm J}$. Its local expression in natural coordinates of $J^1\pi$ is \cite{EMR_96,Sau_89}
$${\rm J}=\left(\d y^a-y^a_\mu\d x^\mu\right)\otimes
\parder{}{y^a_\nu}\otimes\parder{}{x^\nu}\,.
$$
As ${\cal P}=J^1\pi\times_M\bigwedge\nolimits^{k-1}(\Tan^*M)$
is a trivial bundle, this canonical structure 
can be extended to ${\cal P}$ in a natural way. This extension is denoted with the same notation ${\rm J}$, and has the same coordinate expression.
Then, a $m$-multivector field $\bfX\in\vf^m(\mathcal{M})$ is semi-holonomic
if, and only if,
$\inn{\bfX}{\rm J}=0$,
where $\inn{\bfX}{\rm J}$ denotes the natural inner contraction between the tensor fields.
For an Ehresmann connection $\nabla\in\mathcal T^1_1(\mathcal{M})$, this condition is
$\inn{\nabla}{\rm J}=0$, where the tensor contraction is between the contravariant part of $\nabla$ and the covariant part of~${\rm J}$.

\vspace{5pt}

Physical information in field theories is given by {\sl Lagrangian densities}.
A {\sl\textbf{Lagrangian density}} is a $\tau$-semibasic $m$-form $\mathcal{L}\in\df^m({\cal P})$;
hence $\mathcal{L}=L\,\d^mx$, where $L\in\Cinfty({\cal P})$ is the
{\sl\textbf{Lagrangian function}} and $\d^mx$ is also the local expression of the form $\omega:=\tau^*\omega_M$.
Then, the {\sl\textbf{Lagrangian form}} associated with $\mathcal{L}$ is the form
\[
\Theta_{\mathcal{L}}=-\inn{\rm J}\d\mathcal{L}-\mathcal{L}+\d S\in\df^m({\cal P}) \,,
\]
whose expression in natural coordinates reads
\begin{equation}
\label{thetacoor1}
    \Theta_{\mathcal{L}} =-\frac{\partial L}{\partial y^A_\mu}\d y^A\wedge\d^{m-1}x_\mu +\left(\frac{\partial L}{\partial y^A_\mu}y^A_\mu-L\right)\d^m x+\d s^\mu\wedge \d^{m-1}x_\mu\,,
\end{equation}
where the local function
$\displaystyle E_{\mathcal{L}}:=\frac{\partial L}{\partial y^A_\mu}y^A_\mu-L$
is called the {\sl\textbf{Lagrangian energy}} associated with $\mathcal{L}$.

\begin{proposition}
\label{Prop-regLag2}
The Lagrangian form $\Theta_{\mathcal{L}}$ is a special (variational) multicontact form in ${\cal P}$
(and hence $(\Theta_{\mathcal{L}},\omega)$ is a special (variational) multicontact structure)
if, and only if, the Hessian matrix
$\displaystyle (W_{ij}^{\mu\nu})= 
\bigg(\frac{\partial^2L}{\partial y^A_\mu\partial y^B_\nu}\bigg)$
is regular everywhere.
\end{proposition}

\begin{definition}
A Lagrangian function $L\in\Cinfty({\cal P})$ is said to be {\sl\textbf{regular}} if the equivalent
conditions of Proposition \ref{Prop-regLag2} hold.
Otherwise, $L$ is a {\sl\textbf{singular}} Lagrangian.
\end{definition}

It is important to point out that
singular Lagrangians can induce premulticontact structures, but also structures which are neither multicontact nor premulticontact.
Then, when a Lagrangian $L\in\Cinfty({\cal P})$ originates a (pre)multicontact structure,
the triad $({\cal P},\Theta_{\mathcal{L}},\omega)$ is called a {\sl\textbf{(pre)multicontact Lagrangian system}}.

For a multicontact Lagrangian system $({\cal P},\Theta_{\mathcal{L}},\omega)$,
there exists the inverse 
$(W^{AB}_{\mu\nu})$ of the Hessian matrix,
namely $\displaystyle W^{AB}_{\mu\nu}\frac{\partial^2L}{\partial y^B_\nu \partial y^C_\gamma}=\delta^A_C\delta^\gamma_\mu$.
Then, the the Reeb vector fields are
$$
(R_{\mathcal{L}})_\mu=\frac{\partial}{\partial s^\mu}-W^{BA}_{\gamma\nu}\frac{\partial^2L}{\partial s^\mu\partial y^B_\gamma}\,\frac{\partial}{\partial y^A_\nu} \,,
$$
and, for the dissipation form, we obtain that
\begin{equation}
\displaystyle \sigma_{\Theta_{\mathcal{L}}}=-\parder{L}{s^\mu}\,\d x^\mu\,.
\label{sigmaL}
\end{equation}
If $({\cal P},\Theta_{\mathcal{L}})$ is a premulticontact Lagrangian system, the Reeb vector fields are not uniquely determined;
but the local expression of the dissipation form is the same, \eqref{sigmaL}.

Finally, in order to write the Lagrangian field equations we need to introduce the form
$$
\overline\d\Theta_\mathcal{L} =\d\Theta_\mathcal{L}+\sigma_{\Theta_\mathcal{L}}\wedge\Theta_\mathcal{L}\,,
$$
whose local expression is
\begin{equation}
\label{bardtheta}
\bd\Theta_{\mathcal{L}}=
\d\left(-\frac{\partial L}{\partial y^A_\mu}\d y^A\wedge\d^{m-1}x_\mu +\Big(\frac{\partial L}{\partial y^A_\mu}y^A_\mu-L\Big)\d^m x\right)
-\left(\parder{L}{s^\mu}\frac{\partial L}{\partial y^A_\mu}\d y^A
-\parder{L}{s^\mu}\d s^\mu\right)\wedge\d^mx
\,.
\end{equation}

For a (pre)multicontact Lagrangian system $({\cal P},\Theta_\mathcal{L},\omega)$,
the {\sl\textbf{Lagrangian problem}} consists in finding holonomic sections 
$\bm{\psi}_{\mathcal{L}}=j^1\bm\phi\colon M\to{\cal P}$ which are solutions to the Lagrangian field equations which are derived from the {\sl generalized Herglotz Variational Principle} \cite{GLMR_24},
and are stated equivalently (at least locally) as follows:

\begin{enumerate}[{\rm(1)}]
\item  
These holonomic sections are solutions to the {\sl\textbf{(pre)multicontact Lagrangian equations for holonomic sections}}:
\begin{equation}
\label{sect1H}
(j^1\bm\phi)^*\Theta_{\cal L}= 0  \,,\quad
(j^1\bm\phi)^*\inn{X}\bd\Theta_{\cal L}= 0 \,, \qquad \text{for every }\ X\in\vf({\cal P}) \,,
\end{equation}
or, equivalently, for their canonical prolongations $(j^1\bm\phi)^{(m)}$,
\begin{equation}
\label{sect2H}
\inn{(j^1\bm\phi)^{(m)}}(\Theta_{\mathcal{L}}\circ\bm{\psi})=0 \,,\qquad
\inn{(j^1\bm\phi)^{(m)}}(\bd\Theta_{\mathcal{L}}\circ\bm{\psi}) = 0\,.
\end{equation}
\item 
They are the integral sections of holonomic multivector fields $\bfX_{\mathcal{L}} \in\vf^m({\cal P})$ 
which are solutions to the {\sl\textbf{(pre)multicontact Lagrangian equations for multivector fields}}:
\begin{equation}\label{vfH}
\inn{\mathbf{X}_{\mathcal{L}}}\Theta_{\cal L}=0 \,,\qquad \inn{\bfX_{\mathcal{L}}}\bd\Theta_{\cal L}=0 \,,
\end{equation}
where the condition of $\tau$-transversality 
can be imposed simply by asking $\inn{\bfX_{\mathcal{L}}}\omega=1$.
Equations \eqref{vfH} and the $\tau$-transversality condition hold for every multivector field of the equivalence class $\{\bfX_\mathcal{L}\}$.
The $\tau$-transverse locally decomposable multivector fields which are solutions to \eqref{vfH} are called {\sl\textbf{Lagrangian multivector fields}} and, if they are holonomic, are called {\sl\textbf{Euler--Lagrange multivector fields}} (associated with $\mathcal{L}$).
\item
They are the integral sections of holonomic Ehresmann connections 
$\nabla_{\mathcal{L}}\in\mathcal{T}^1_1({\cal P})$ which are solutions to the {\sl\textbf{(pre)multicontact Lagrangian equations for Ehresmann connections}}:
\begin{equation}
\label{EcL}
\inn{\nabla_{\mathcal{L}}}\Theta_{\mathcal{L}}=(m-1)\Theta_{\mathcal{L}} \,,\qquad
\inn{\nabla_{\mathcal{L}}}\bd\Theta_{\mathcal{L}}=(m-1)\bd\Theta_{\mathcal{L}}  \,.
\end{equation}
Ehresmann connections which are solutions to these equations are called {\sl\textbf{Lagrangian connections}} and, if they are holonomic, are called {\sl\textbf{Euler--Lagrange connections}}.
\end{enumerate}

In a natural chart of coordinates of ${\cal P}$, 
a $\tau$-transverse and locally decomposable $m$-multivector field
satisfying $\inn{\bfX_{\mathcal{L}}}\omega=1$, has the local expression
$$\displaystyle
{\bf X}_{\mathcal{L}}= \bigwedge_{\mu=1}^m
\bigg(\parder{}{x^\mu}+(X_{\mathcal{L}})_\mu^A\frac{\displaystyle\partial}{\displaystyle
\partial y^A}+(X_{\mathcal{L}})_{\mu\nu}^A\frac{\displaystyle\partial}{\displaystyle\partial y^A_\nu}+(X_{\mathcal{L}})_\mu^\nu\,\frac{\partial}{\partial s^\nu}\bigg) \,;
$$
and the Ehresmann connection $\nabla_{\mathcal{L}}$
associated with the class $\{{\bf X}_{\mathcal{L}}\}$ is,
$$\displaystyle
\nabla_{\mathcal{L}}= \d x^\mu\otimes
\bigg(\derpar{}{x^\mu}+(X_{\mathcal{L}})_\mu^A\frac{\displaystyle\partial}{\displaystyle
\partial y^A}+(X_{\mathcal{L}})_{\mu\nu}^A\frac{\displaystyle\partial}{\displaystyle\partial y^A_\nu}+(X_{\mathcal{L}})_\mu^\nu\,\frac{\partial}{\partial s^\nu}\bigg) \,,
$$
and, if $\bfX_{\mathcal{L}}$ and $\nabla_{\mathcal{L}}$ are {\sc sopde}s, they are semi-holonomic and
\begin{equation}
\label{semihol}
y^A_\mu=(X_{\mathcal{L}})_\mu^A \,.
\end{equation}
Then, bearing in mind the local expressions\eqref{thetacoor1} and \eqref{bardtheta},
equations \eqref{vfH} or \eqref{EcL} lead to: 
\begin{align}
(X_{\mathcal{L}})_\mu^\mu&= L \,, \nonumber
\\
\frac{\partial L}{\partial y^A}- \parderr{L}{x^\mu}{y_\mu^A}
-\frac{\partial^2L}{\partial y^B \partial y^A_\mu}y_\mu^B
-\frac{\partial^2L}{\partial s^\nu\partial y^A_\mu}(X_{\mathcal{L}})_\mu^\nu
-\frac{\partial^2L}{\partial y^B_\nu\partial y^A_\mu}(X_{\mathcal{L}})_{\mu\nu}^B
&=-\frac{\partial L}{\partial s^\mu}
\frac{\partial L}{\partial y^A_\mu} \,.
\label{ELeqmvf}
\end{align}
Observe that equations \eqref{ELeqmvf} are compatible when $L$ is regular since the Hessian matrix 
$\displaystyle\bigg(\frac{\partial^2L}{\partial y^B_\nu\partial y^A_\mu}\bigg)$ is regular everywhere (the solution is not unique unless $m=1$), but they could be incompatible in the singular case.
In addition, if $L$ is a regular Lagrangian, the semi-holonomy condition \eqref{semihol} always holds.
Finally, for the holonomic integral sections $\displaystyle\bm{\psi}(x^\nu)=\Big(x^\mu,y^A(x^\nu),\parder{y^A}{x^\mu}(x^\nu),s^\mu(x^\nu)\Big)$ of $\bfX_{\mathcal{L}}$
and $\nabla_{\mathcal{L}}$, these last equations transform into,
\begin{align}
    \label{actioneqs}
     \parder{s^\mu}{x^\mu}&=L\circ{\bm{\psi}} \,, 
     \\
    \label{ELeqs2}
    \frac{\partial}{\partial x^\mu}
    \left(\frac{\displaystyle\partial L}{\partial
    y^B_\mu}\circ{\bm{\psi}}\right)&=
    \left(\frac{\partial L}{\partial y^B}+
    \displaystyle\frac{\partial L}{\partial s^\mu}\displaystyle\frac{\partial L}{\partial y^B_\mu}\right)\circ{\bm{\psi}} \,,
\end{align}
which are the coordinate expression of the Lagrangian equations
\eqref{sect1H} or \eqref{sect2H} for holonomic sections.
s \eqref{ELeqs2}
are called the {\sl\textbf{Herglotz--Euler--Lagrange field equations}}.

When $L$ is not regular and $({\cal P},\Theta_{\mathcal{L}},\omega)$
is a premulticontact system, 
the field equations \eqref{vfH} and \eqref{EcL} have no
solutions everywhere on ${\cal P}$, in general. 
In the most favourable situations,
they have solutions on a submanifold of ${\cal P}$ which is obtained by applying a suitable constraint algorithm.
However, solutions to equations \eqref{vfH} and \eqref{EcL}
are not necessarily {\sc sopde}s and then, if they are integrable, 
their integral sections are not necessarily holonomic;
so this requirement must be imposed as an additional condition.
Thus, the final step consists in finding the maximal submanifold 
${\cal S}_f$ of ${\cal P}$ where there are 
Euler--Lagrange multivector fields and connections which are
solutions to the premulticontact Lagrangian field equations on ${\cal S}_f$ and are tangent to ${\cal S}_f$.

\subsection{Multicontact Hamiltonian formalism}
\label{multcontHam}

Consider the bundle $\pi\colon E\to M$. 
Let ${\cal M}\pi\equiv\bigwedge\nolimits_2^m\Tan^*E$ denote the bundle of $m$-forms on
$E$ vanishing by contraction with two $\pi$-vertical vector fields which, in field theories, is called the {\sl extended multimomentum bundle}.
It has natural coordinates $(x^\nu,y^A,p^\nu_A,p)$
adapted to the bundle structure ${\cal M}\pi\to E\to M$, and such that
$\omega_M=\d^mx$; so $\dim\, {\cal M}\pi=nm+n+m+1$.
Consider also the quotient manifold
$J^{1*}\pi={\cal M}\pi/\pi^*\bigwedge\nolimits^m\Tan^*M$
($\pi^*\bigwedge\nolimits^m\Tan^*M$ is the bundle of $\pi$-basic $m$-forms on $E$),
which is called the {\sl restricted multimomentum bundle}.
Its natural coordinates are $(x^\mu,y^A,p_A^\mu)$, and so 
$\dim\,J^{1*}\pi=nm+n+m$.

Then, for the Hamiltonian formalism of action-dependent field theories,
in the regular case, consider the {\sl multimomentum phase bundles}
$$
\widetilde{\cal P}={\cal M}\pi\times_M\bigwedge\nolimits^{m-1}\Tan^* M
\ , \qquad
{\cal P}^* =J^{1*}\pi\times_M\bigwedge\nolimits^{m-1}\Tan^*M \,,
$$
which have natural coordinates $(x^\mu, y^A,p_A^\mu,p,s^\mu)$ and$(x^\mu,y^A,p_A^\mu,s^\mu)$, respectively. 
We have the natural projections depicted in the diagram:
\begin{equation}\label{firstdiag}
\begin{tikzcd}[column sep=huge, row sep=huge]
    & \widetilde{\mathcal{P}} = \mathcal{M}\pi \times_M\bigwedge\nolimits^{m-1}\Tan^*M
    \arrow[dl, "\widetilde{\rho}_1", swap, bend right=10]
    \arrow[ddl, "\widetilde{\rho}", swap, bend right=10]
    \arrow[d, "\widetilde{\mathfrak{p}}", swap]
    \arrow[ddr, "\widetilde\tau_1", bend left=20] \\
    \mathcal{M}\pi 
    \arrow[d, "\mathfrak{p}", swap]
    & {\cal P}^*=J^{1*}\pi\times_M\bigwedge\nolimits^{m-1}\Tan^*M
    \arrow[u, "\mathbf{h}", swap, bend right]
    \arrow[dl, "\varrho", swap] 
    \arrow[ddd, "\bar\tau"] 
    \arrow[dr, "\bar\tau_1"] 
    \arrow[ddl, "\varrho_E"] \\
    J^{1*}\pi
    \arrow[d, "\kappa^1", swap]
    \arrow[ddr, "\bar\kappa^1", crossing over, yshift = 2.0]
    & & \bigwedge\nolimits^{m-1}\Tan^*M
    \arrow[ddl, "\tau_\circ"] \\
    E
    \arrow[dr, "\pi", swap, yshift = -1.0] \\
    & M
\end{tikzcd}
\end{equation}

Since ${\cal M}\pi$ and $\bigwedge\nolimits^{m-1}\Tan^*M$ are bundles of forms, they have canonical structures, namely their {\sl tautological forms}
$\Theta\in\df^m({\cal M}\pi)$ and $\theta\in \df^{m-1}(\bigwedge\nolimits^{m-1}\Tan^*M)$.
In particular, if $\mathfrak{p}_1\colon{\cal M}\pi\to E$ is the canonical projection,
$(y, \alpha)$ is a point in $\bigwedge\nolimits^m_2 \mathrm{T}^*E$, so that $y \in E$ and $\alpha \in \bigwedge\nolimits^m_2 \mathrm{T}^*_y E$; then, for every $X_1,\dotsc, X_m \in \mathrm{T}_{(y, \alpha)}(\mathcal{M}\pi)$, the so-called {\sl\textbf{Liouville form}} $\Theta \in \Omega^m(\mathcal{M\pi})$ is defined as:
    \[
\Theta((y, \alpha); X_1, \dotsc, X_m) := \alpha (y; \mathrm{T}_{(y, \alpha)} \mathfrak{p}_1(X_1), \dotsc, \mathrm{T}_{(y, \alpha)} \mathfrak{p}_1(X_m)) \, .
    \]
and we naturally define $\Omega = \mathrm{d} \Theta \in \Omega^{m+1}(\mathcal{M}\pi)$. 
Similarly, we can define the form $\theta \in \bigwedge\nolimits^{m-1} \mathrm{T}^* M$.
Their local expressions are
\begin{equation}
\label{canformmulticot}
\Theta=p_A^\mu\d y^A\wedge\d^{m-1}x_\mu+p\,\d^m x
\,,\quad
\theta=s^\mu\,\d^{m-1}x_\mu \,.
\end{equation}

\begin{definition}
The {\sl\textbf{canonical (special) multicontact form}} of $\widetilde{\cal P}$ is
\begin{equation}
\label{canoncontact}
\widetilde\Theta:= -\widetilde\rho_1^{\,*}\Theta+\d(\widetilde\tau_1^{\,*}\theta) \,. 
\end{equation}
\end{definition}

Its local expression in natural coordinates is
$$
\widetilde\Theta=
-p_A^\mu\d y^A\wedge\d^{m-1}x_\mu-p\,\d^m x+\d s^\mu\wedge \d^{m-1}x_\mu\ . 
$$

Next, we construct the {\sl canonical multicontact Hamiltonian system} associated with a multicontact Lagrangian system $({\cal P},\Theta_\mathcal{L},\omega)$, with $\mathcal{L}=L\,\omega$.
First, denote $\mathcal{F}L_s\colon J^1\pi\to J^{1*}\pi$ the Legendre map associated with the restriction of the Lagrangian function $L\in\Cinfty({\cal P})$
to the fibers of the projection $\tau_1$, which is denoted $L_s\in\Cinfty(J^1\pi)$ (see diagram  \eqref{diagrammulticoLag}).
Then, the {\sl\textbf{restricted Legendre map}} associated with the Lagrangian function $L\in\Cinfty({\cal P})$
is the map
${\cal FL}\colon {\cal P}\to {\cal P}^*$
given by ${\cal FL}:=(\mathcal{F}L_s,{\rm Id}_{\bigwedge\nolimits^{m-1}\Tan^*M})$;
which is locally given by
$$\displaystyle
{\cal FL}(x^\mu,y^A,y^A_\mu,s^\mu)=\Big(x^\mu,y^A,\frac{\partial L}{\partial y^A_\mu},s^\mu\Big) \,. $$
Similarly, the {\sl\textbf{extended Legendre map}} associated with $L$ is the map 
$\widetilde{\mathcal{FL}}\colon{\cal P}\to\widetilde{\cal P}$ given by $\widetilde{\mathcal{FL}}:=(\widetilde{\mathcal{F}L}_s,{\rm Id}_{\bigwedge\nolimits^{m-1}\Tan^*M})$,
and its local expression is
$$\widetilde{\mathcal{FL}}(x^\mu,y^A,y^A_\mu,s^\mu)=\Big(x^\mu,y^A,\frac{\partial L}{\partial y^A_\mu},L-y^A_\mu\parder{L}{y^A_\mu},s^\mu\Big)\ . $$
It is proved that $L$ is a regular Lagrangian function if, and only if,
the Legendre map ${\cal FL}$ is a local diffeomorphism;
in particular, 
$L$ is said to be {\sl\textbf{hyperregular}}
when ${\cal FL}$ is a global diffeomorphism.

\subsubsection*{The (hyper)regular case}

First, consider the case where $L$ is a hyperregular Lagrangian function. We can define a section ${\bf h}:=\widetilde{\mathcal{FL}}\circ\mathcal{FL}^{-1}\colon {\cal P}^*\to\widetilde{\cal P}$ of the projection $\widetilde{\mathfrak{p}}$
(see Diagram \eqref{firstdiag}),
which is locally determined by a function $H\in\Cinfty(U)$, $U\subset{\cal P}^*$,
such that ${\bf h}(x^\mu,y^A,p^\mu_A,s^\mu)=(x^\mu,y^A,p^\mu_A,p=-H(x^\nu,y^B,p^\nu_B,s^\nu),s^\mu)$.
The elements ${\bf h}$ and $H$ are called
a {\sl\textbf{Hamiltonian section}} and its associated {\sl\textbf{Hamiltonian function}}.
Then, denoting $\overline{S}\equiv\overline{\tau}_1^*\theta$, the {\sl\textbf{Hamiltonian form}} associated with ${\bf h}$ is defined by
\begin{equation}
\label{thetaHcoor}
\Theta_{\cal H}:=\mathbf{h}^*\widetilde{\Theta}=
-(\widetilde\varrho_1\circ{\bf h})^*\Theta+\d\overline{S} \,. 
\end{equation}
It is a special variational multicontact form whose local expression is
\begin{equation}
\label{thetacoor2}
\Theta_{\cal H}=
-p_A^\mu\d y^A\wedge\d^{m-1}x_\mu+H\,\d^m x+\d s^\mu\wedge \d^{m-1}x_\mu\ ,
\end{equation}
and the dissipation form is expressed as
\begin{equation}
\label{sigmaH}
\sigma_{\cal H}=\parder{H}{s^\mu}\,\d x^\mu\,.
\end{equation}
The triad $({\cal P}^*,\Theta_{\cal H},\omega=(\bar{\tau}\circ\widetilde{\mathfrak{p}})^*\omega_M)$
is called the {\sl\textbf{canonical multicontact Hamiltonian system}} associated with the multicontact Lagrangian system
$({\cal P},\Theta_\mathcal{L},\omega)$.
Bearing in mind the coordinate expressions \eqref{thetacoor1}, \eqref{sigmaL}, \eqref{thetacoor2}, and \eqref{sigmaH},
we obtain $\mathcal{FL}^*\Theta_\mathcal{H}=\Theta_{\mathcal{L}}$ and
$\mathcal{FL}^*\bd\Theta_\mathcal{H}=\bd\Theta_{\mathcal{L}}$.

If $\mathcal{L}$ is regular, this construction is local.

For this multicontact Hamiltonian system $({\cal P}^*,\Theta_{\cal H},\omega)$, 
the field equations can be stated in three equivalent ways (at least locally). 
Thus, the {\sl\textbf{Hamiltonian problem}} consists of finding sections $\bm{\psi}_{\mathcal{H}}\colon M\to{\cal P}^*$ such that:
\begin{enumerate}[{\rm(1)}]
\item  
They are solutions to the {\sl\textbf{(pre)multicontact Hamilton--de Donder--Weyl equations for sections}}:
\begin{equation}
\label{sect1H0}
\bm{\psi}_{\mathcal{H}}^*\Theta_{\cal H}= 0  \,, \quad
\bm{\psi}_{\mathcal{H}}^*\inn{Y}\bd\Theta_{\cal H}= 0 \,, \qquad \text{for every }\ Y\in\vf({\cal P}^*) \, ,
\end{equation}
or, equivalently, for their canonical prolongations,
\begin{equation}
\label{sect2H0}
\inn{\bm{\psi}_{\mathcal{H}}^{(m)}}(\Theta_{\cal H}\circ\bm{\psi}_{\mathcal{H}})=0 \,,\qquad
\inn{\bm{\psi}_{\mathcal{H}}^{(m)}}(\bd\Theta_{\cal H}\circ\bm{\psi}_{\mathcal{H}}) = 0\,.
\end{equation}
\item 
They are the integral sections of $\overline\tau$-transverse and integrable multivector fields ${\bf X}_{\cal H}\in\vf^m({\cal P}^*)$ which are solutions to the {\sl\textbf{(pre)multicontact Hamilton--de Donder--Weyl for multivector fields}}:
\begin{equation}
\label{vfH2}
\inn{{\bf X}_{\cal H}}\Theta_{\cal H}=0 \,, \qquad \inn{\bfX_{\cal H}}\bd\Theta_{\cal H}=0 \,.
\end{equation}
Equations \eqref{vfH2} and the $\overline\tau$-transversality condition hold for every multivector field of the equivalence class $\{\bfX_\mathcal{H}\}$,
and the transversality condition can be imposed by asking $\inn{\bfX_\mathcal{H}}\omega=1$.
\item
They are the integral sections of Ehresmann connections
$\nabla_{\cal H}\in\mathcal{T}^1_1({\cal P}^*)$
which are solutions to the {\sl\textbf{(pre)multicontact Hamilton--de Donder--Weyl equations for Ehresmann connections}}:
\begin{equation}
\label{EcH}
\inn{\nabla_{\cal H}}\Theta_{\cal H}=(m-1)\Theta_{\cal H} \,,\qquad
\inn{\nabla_{\cal H}}\bd\Theta_{\cal H}=(m-1)\bd\Theta_{\cal H}  \,.
\end{equation}
\end{enumerate}

In natural coordinates, a $\overline\tau$-transverse, locally decomposable multivector field and an Ehresmann connection in $\mathcal{P}^*$ are:
\begin{align}
\label{Hammv}
{\bf X}_{\cal H} &=
\bigwedge_{\mu=0}^{m-1}\Big(\parder{}{x^\mu}+ (X_{\cal H})^A_\mu\frac{\partial}{\partial y^A}+
(X_{\cal H})_{\mu A}^\nu\frac{\partial}{\partial p_A^\nu}+(X_{\cal H})_\mu^\nu\parder{}{s^\nu}\Big) \,, \\
\nabla_{\cal H} &= \d x^\mu\otimes\Big(\derpar{}{x^\mu}+ (X_\mu)^A\frac{\partial}{\partial y^A}+
(X_\mu)_A^\nu\frac{\partial}{\partial p_A^\nu}+(X_\mu)^\nu\derpar{}{s^\nu}\Big) \,.
\end{align}
If they are solutions to equations \eqref{vfH2} and \eqref{EcH},
keeping in mind the local expression \eqref{thetacoor2},
these field equations lead to
\begin{gather}
(X_{\cal H})_\mu^\mu = 
p_A^\mu\,\frac{\partial H}{\partial p^\mu_A}-H \,,\\
(X_{\cal H})^A_\mu=\frac{\partial H}{\partial p^\mu_A} \,,  \qquad
(X_{\cal H})_{\mu a}^\mu= 
-\Big(\frac{\partial H}{\partial y^A}+ p_A^\mu\,\frac{\partial H}{\partial s^\mu}\Big) \,.  
\end{gather}
Then, the integral sections $\bm{\psi}(x^\nu)=(x^\mu,y^A(x^\nu),p^\mu_A(x^\nu),s^\mu(x^\nu))$
of the integrable solutions to \eqref{vfH2} and \eqref{EcH}
are the solutions to the equations \eqref{sect1H0} and \eqref{sect2H0} which read as
\begin{gather}
\frac{\partial s^\mu}{\partial x^\mu} = \Big(p_A^\mu\,\frac{\partial H}{\partial p^\mu_A}-H\Big)\circ\bm{\psi}_{\mathcal{H}}\,, 
\label{actionHameqs} \\ 
\frac{\partial y^A}{\partial x^\mu}= \frac{\partial H}{\partial p^\mu_A}\circ\bm{\psi}_{\mathcal{H}} \,,  \qquad
\frac{\partial p^\mu_A}{\partial x^\mu} = 
-\Big(\frac{\partial H}{\partial y^A}+ p_A^\mu\,\frac{\partial H}{\partial s^\mu}\Big)\circ\bm{\psi}_{\mathcal{H}} \,.
\label{HHDWeqs}
\end{gather}
Equations \eqref{HHDWeqs} are called the {\sl\textbf{Herglotz--Hamilton--de Donder--Weyl equations}}
for action-dependent field theories.
These equations are compatible in ${\cal P}^*$.

\subsubsection*{The singular (almost-regular) case}

For singular Lagrangians, the existence of
an associated Hamiltonian formalism is not assured, in general,
unless some minimal regularity conditions are assumed:

\begin{definition}
A singular Lagrangian $L\in\Cinfty({\cal P})$ is {\sl\textbf{almost-regular}} if
\ (i) ${\cal P}_\circ^*:= {\cal FL}({\cal P})$
is a submanifold of ${\cal P}^*$,
\ (ii) ${\cal F}L$ is a submersion onto its image,
and \ (iii)
the fibers ${\cal FL}^{-1}({\rm p})$ are connected submanifolds of ${\cal P}$, for every ${\rm p}\in {\cal P}_\circ^*$.
\end{definition}

We have that $L\in\Cinfty({\cal P})$ is an almost-regular Lagrangian on ${\cal P}$
if, and only if, $L_s\in\Cinfty(J^1\pi)$ is an almost-regular Lagrangian on $J^1 \pi$.
Then, we consider the submanifolds
\begin{gather}
    \widetilde P_\circ=\widetilde{\mathcal{F}L}_s(J^1\pi)\,,\qquad
    \widetilde{\cal P}_\circ=\widetilde P_\circ\times\bigwedge\nolimits^{m-1}\Tan^*M=
    \widetilde{\cal FL}({\cal P})\,,\\
    P_\circ=\mathcal{F}L_s(J^1\pi)\,,\qquad {\cal P}^*_\circ=P_\circ\times\bigwedge\nolimits^{m-1}\Tan^*M={\cal FL}({\cal P})\,,
\end{gather}
and we have the following diagram:
\begin{equation}
\begin{tikzcd}[column sep=huge, row sep=large]
    & \mathcal{P}^* = J^{1*}\pi\times_M\bigwedge^{m-1}\Tan^*M\\
    & {\cal P}_\circ^* = P_\circ\times_M\bigwedge\nolimits^{m-1}\Tan^*M
    \arrow[u, hook, "\jmath_\circ"]
    \arrow[dl, "\varrho_\circ", swap] 
    \arrow[ddd, "\widetilde\tau_\circ"] 
    \arrow[dr, "\upsilon"] \\
    P_\circ
    \arrow[d, "\kappa^1_\circ", swap]
    \arrow[ddr, "\bar\kappa^1_\circ", yshift= 2.0]
    & & \bigwedge\nolimits^{m-1}\Tan^*M
    \arrow[ddl, "\tau_\circ"] \\
    E
    \arrow[dr, "\pi", swap, yshift= -1.0] \\
    & M
\end{tikzcd}
\end{equation}

The restriction of $\widetilde{\mathfrak{p}}$ to $\widetilde{\cal P}_\circ$, denoted
$\mathfrak{p}_\circ\colon\widetilde{\cal P}_\circ\to{\cal P}^*_\circ$, is a diffeomorphism
(the proof is the same as for multisymplectic Hamiltonian singular field theories \cite{LMM_96a,LMM_96}).
Then, we can take ${\bf h}_\circ:=\mathfrak{p}_\circ^{-1}$
and construct the Hamiltonian sections
$\mathbf{h}_\circ:=\widetilde{\jmath_\circ}\circ\mathbf{h}_\circ\colon \mathcal{P}_\circ^*\to\widetilde{\mathcal{P}}_\circ$,
and $\widetilde{\mathbf{h}}\colon \mathcal{P}^*\to\widetilde{\mathcal{P}}$ defined by
$\widetilde{\bf h}_\circ=\mathbf{h}_\circ \circ\jmath_\circ$.
The diagram corresponding to the Hamiltonian formalism is now the following:
\begin{equation}\label{diaghamsec0}
\begin{tikzcd}[column sep=huge, row sep=huge]
    & \widetilde{\mathcal{P}}_\circ
    \arrow[r, "\widetilde{\jmath_\circ}"]
    \arrow[d, "\mathfrak{p}_\circ", swap]
    & \widetilde{\mathcal{P}}
    \arrow[d, "\widetilde{\mathfrak{p}}", swap] \\
    \mathcal{P}
    \arrow[ur, "\widetilde{\mathcal{FL}_\circ}"]
    \arrow[r, "\mathcal{FL}_\circ", swap]
    & \mathcal{P}_\circ^*
    \arrow[ur, "\widetilde{\mathbf{h}}_\circ", swap]
    \arrow[r, "\jmath_\circ", swap]
    \arrow[u, "\mathbf{h}_\circ", swap, bend right]
    & \mathcal{P}^*
    \arrow[u, "\widetilde{\mathbf{h}}", swap, bend right]
\end{tikzcd}
\end{equation}

Then, we define the Hamilton--Cartan $m$-form
$\Theta^\circ_{\bf h}:=\widetilde{\bf h}_\circ^*\widetilde\Theta\in\df^m({\cal P}^*_\circ)$
and, taking $\overline{S}_\circ=\jmath_\circ^*\overline{S}$ we can construct the form
$$
\Theta^\circ_{\cal H}=-\Theta_{\bf h}^\circ+\d\overline{S}_\circ \in\df^m({\cal P}_\circ^*)\,,
$$
whose local expression is
$$
\Theta_{\cal H}^\circ=
\jmath_\circ^{\,*}(-p_i^\nu\d y^i\wedge\d^{m-1}x_\mu)+H_\circ\,\d^mx+\d s^\mu\wedge\d x^{m-1}x_\mu \,,
$$
where $H_\circ\in\Cinfty(P_\circ)$ is the Hamiltonian
function such that $E_\mathcal{L} = \mathcal{F}L_{s\circ}^{\ *}\,H_\circ$.
Denoting by
$\mathcal{FL}_\circ\colon{\cal P}\to{\cal P}^*$ the restriction of
$\mathcal{FL}$ on $P_\circ$, given as
$\mathcal{FL}={\rm j}_\circ\circ\mathcal{FL}_\circ$;
we have $\Theta_{\mathcal{L}}=\mathcal{FL}_\circ^{\ *}\,\Theta_{\cal H}^\circ$
and $\mathcal{FL}_\circ^*\bd\Theta_\mathcal{H}^\circ=\bd\Theta_{\mathcal{L}}$.

The pair $(\Theta_{\cal H}^\circ,\omega)$ is a premulticontact
structure if, and only if,
$(\Theta_{\mathcal{L}},\omega)$ is a premulticontact structure.
In that case, the triple $({\cal P}_\circ^*,\Theta_{\cal H}^\circ,\omega)$ 
is the {\sl \textbf{premulticontact Hamiltonian system}}
associated to the premulticontact Lagrangian system $({\cal P},\Theta_{\mathcal{L}},\omega)$.

\begin{remark}
In the premulticontact case,
for the premulticontact Hamiltonian system $({\cal P}_\circ^*,\Theta_{\cal H}^\circ,\omega)$, 
the field equations for sections $\bm{\psi}_{\mathcal{H}}^\circ\colon M\to{\cal P}_\circ^*$, 
multivector fields ${\bf X}_{{\cal H}_\circ}\in\vf^m({\cal P}_\circ^*)$, 
and Ehresmann connections $\nabla_{{\cal H}_\circ}\in\mathcal{T}^1_1({\cal P}_\circ^*)$
are like \eqref{sect1H0}, \eqref{sect2H0}, \eqref{vfH2}, and \eqref{EcH}
with $\Theta_{\cal H}^\circ$ instead of $\Theta_{\cal H}$;
that is,
\begin{align}
\label{sect1Hb0}
& \bm{\psi}_\circ^*\Theta_{\cal H}^\circ= 0  \,, &&
\bm{\psi}_\circ^*\innp{Y_\circ}\overline{\Omega}_{\cal H}^\circ= 0 \,, && \text{for every }\ Y_\circ\in\vf({\cal }P^*_\circ) \,, \\
\label{vfHb0}
&\innp{\bfX_{{\cal H}_\circ}}\Theta_{\cal H}^\circ=0 \,, && \innp{\bfX_{{\cal H}_\circ}}\overline{\Omega}_{\cal H}^\circ=0 \, , \\
\label{EcH00}
&\inn{\nabla_{{\cal H}_\circ}}\Theta_{\cal H}^\circ=(m-1)\Theta_{\cal H}^\circ \,,&&
\inn{\nabla_{{\cal H}_\circ}}\bd\Theta_{\cal H}^\circ=(m-1)\bd\Theta_{\cal H}^\circ \,.
\end{align}
In general, these equations are not compatible on ${\cal P}_\circ^*$ 
and a constraint algorithm must be implemented to find a final constraint submanifold ${\cal P}_f^*\hookrightarrow{\cal P}_\circ^*$
where there are integrable
multivector fields and connections solution to the Hamiltonian field equations, which are tangent to ${\cal P}_f^*$.
\end{remark}

\section{Skinner--Rusk formulation of multicontact field theories}
\label{skinner--Rusk}

\subsection{Geometric framework}

Now, we develop the Skinner--Rusk formulation for the multicontact setting of action-dependent field theories.
It is based on the corresponding unified formulations for multisymplectic field theories \cite{LMM_03,ELMMR_04,PR_15} and for contact
and $k$-contact systems \cite{LGLRR_20,GRR_22,RT_23}.

\begin{definition}
The {\sl\textbf{extended}} and {\sl \textbf{restricted jet-multimomentum bundles}}
(which, sometimes, are also called the {\sl\textbf{extended}} and {\sl \textbf{restricted Pontryagin bundles}})
are
\[\mathcal{W}:=(J^1 \pi \times_E \mathcal{M}\pi) \times_M \bigwedge\nolimits^{m-1} \mathrm{T}^*M\,, \qquad \mathcal{W}_r:= (J^1\pi\times_E J^{1*}\pi) \times_M \bigwedge\nolimits^{m-1} \mathrm{T}^*M\,,\]
\end{definition}

They are endowed with natural coordinates $(x^\mu, y^A, y^A_\mu, p^\mu_A, p, s^\mu)$ and $(x^\mu, y^A, y^A_\mu, p^\mu_A, s^\mu)$, respectively. 
There is a natural projection between the extended and restricted bundles, $\mu_\mathcal{W}\colon \mathcal{W} \rightarrow \mathcal{W}_r$, and also the projections induced by the bundle structure:
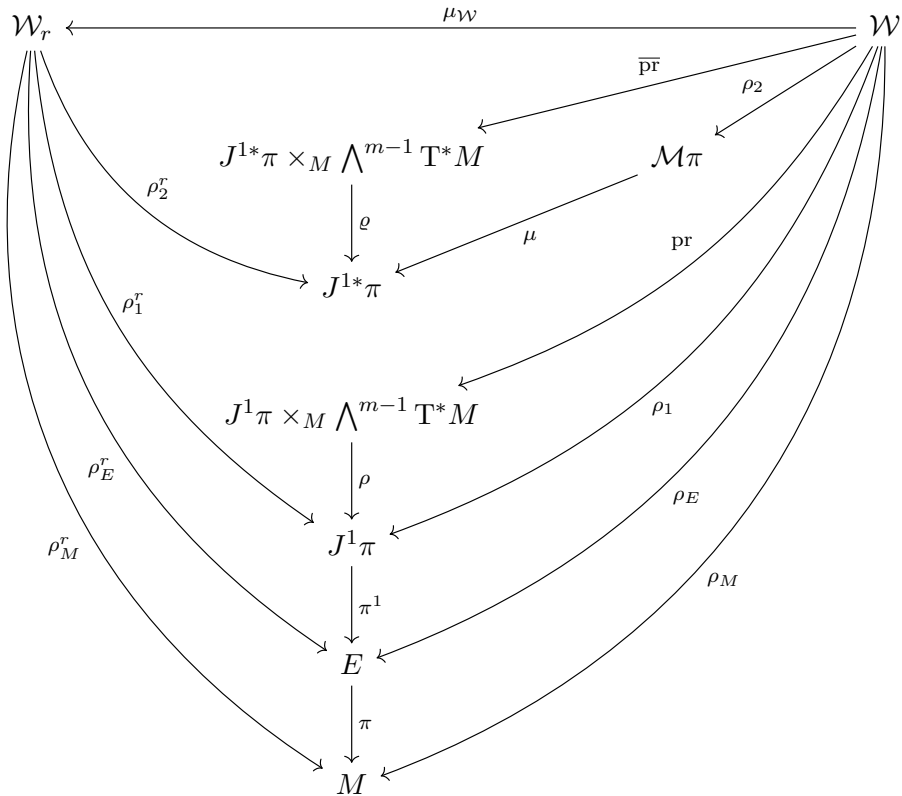
\begin{equation}
\begin{tikzcd}[column sep=huge, row sep=large]
    \mathcal{W}_r 
    \arrow[ddr, "\rho_2^r", bend right]
    \arrow[ddddr, "\rho_1^r", bend right=25]
    \arrow[dddddr, "\rho_E^r", swap, bend right=30, pos=0.6]
    \arrow[ddddddr, "\rho_M^r", swap, bend right=35, pos=0.6]
    & & & \mathcal{W}
    \arrow[lll, "\mu_\mathcal{W}", swap]
    \arrow[dll, "\overline{\pr}", swap]
    \arrow[dl, "\rho_2", pos=0.6, swap]
    \arrow[llddd, "\pr", swap, bend left=20]
    \arrow[lldddd, "\rho_1", bend left=25, pos=0.6]
    \arrow[llddddd, "\rho_E", bend left=30, pos=0.6]
    \arrow[lldddddd, "\rho_M", bend left=35, pos=0.6]
    \\
    & J^{1*}\pi\times_M\bigwedge^{m-1}\cT M
    \arrow[d, "\varrho"]
    & \mathcal{M}\pi
    \arrow[dl, "\mu"]
    \\
    & J^{1*}\pi
    \\
    & J^1\pi\times_M\bigwedge^{m-1}\cT M 
    \arrow[d, "\rho"]
    \\
    & J^1\pi
    \arrow[d, "\pi^1"]
    \\
    & E 
    \arrow[d, "\pi"]
    \\
    & M
\end{tikzcd}
\end{equation}

Moreover, the bundle $\mathcal{W}$ is endowed with the following canonical structures:

    \begin{definition}
The {\sl\textbf{coupling $m$-form}}
in $\mathcal{W}$, is a $m$-form along $\rho_M$ denoted by $\mathcal{C}$ and defined as follows: for every $\bar{y} \in J^1_y\pi$ where $\bar{\pi}^1(\bar{y}) = \pi(y) = x \in M$, and for every $\mathbf{p} \in \mathcal{M}_y\pi$, let $w := (\bar{y}, \mathbf{p}) \in \mathcal{W}_y$, then
\[ \mathcal{C}(w) := (\mathrm{T}_x \phi)^* \mathbf{p}\]
where $\phi\colon M \rightarrow E$ is a section satisfying $j^1\phi(x) = \bar{y}$.

We commit an abuse of notation and denote also by $\mathcal{C}\in\df^m(\mathcal{W})$, the $\rho_M$-semibasic form associated with this coupling form. Then, there exists a function $C\in\Cinfty (\mathcal{W})$ such that $\mathcal{C} = C(\rho^*_M\omega)$, and 
$$\mathcal{C}(w)= (p + p^\mu_A y^A_\mu) \, \d ^m x\, .$$
    \end{definition}
    
Observe that, really, the coupling form is just a canonical structure on the subbundle  $J^1 \pi \oplus \mathcal{M}\pi$ of the trivial bundle $\mathcal{W}:=(J^1 \pi \times_E \mathcal{M}\pi) \times_M \bigwedge\nolimits^{m-1} \mathrm{T}^*M$.

Now, using the canonical tautological forms $\Theta \in \Omega^m(\mathcal{M\pi})$ and $\theta \in \bigwedge\nolimits^{m-1} \mathrm{T}^* M$,
we introduce:
    
\begin{definition}
The {\sl\textbf{canonical $m$-form}} in $\mathcal{W}$, $\Theta_{\mathcal{W}}\in\df^m(\mathcal{W})$,
is defined as 
$$\Theta_{\mathcal{W}} := -\rho_2^*\Theta + \mathrm{d}(\bar{\tau}_2^* \theta) \, ,
$$
where the diagram below shows the projections used:
\begin{equation}
\begin{tikzcd}[column sep=huge, row sep=huge]
    \mathcal{W}_r
    \arrow[ddr, "\overline{\tau}_2^r", bend left=20, pos=0.75]
    &&
    \mathcal{W}
    \arrow[ll, "\mu_{\mathcal{W}}", swap]
    \arrow[d, "\overline{\pr}"]
    \arrow[dll, "\pr", crossing over]
    \arrow[ddl, "\overline{\tau}_2", bend right=20, yshift = -1.0]
    \\
    \mathcal{P} = J^1\pi\times_M\bigwedge^{m-1}\cT M
    \arrow[d, "\rho", swap]
    \arrow[dr, "\tau^1", swap]
    &&
    \mathcal{P}^* = J^{1*}\pi\times_M\bigwedge^{m-1}\cT M
    \arrow[d, "\varrho"]
    \arrow[dl, "\overline{\tau}^1"]
    \\
    J^1\pi
    &
    \bigwedge^{m-1}\cT M
    &
    J^{1*}\pi
\end{tikzcd}
\end{equation}
\end{definition}
        
In coordinates,
\[ \Theta_\mathcal{W} = -p_A^\mu \, \mathrm{d}y^A \wedge \mathrm{d} ^{m-1}x_\mu - p \,  \mathrm{d}^mx + \mathrm{d}s^\mu \wedge \mathrm{d}^{m-1}x_\mu \, , \]
and 
\[ \d\Theta_\mathcal{W} = -\d p_A^\mu \wedge \mathrm{d}y^A \wedge \mathrm{d} ^{m-1}x_\mu - \d p \wedge \mathrm{d}^mx \, , \]

\begin{remark}
From the above local expression, we obtain,
\begin{align}
\ker\Theta_{\mathcal{W}}=&\left\langle \frac{\partial}{\partial y^A}+p^\mu_A\frac{\partial}{\partial s^\mu},\frac{\partial}{\partial p_A^\mu},\frac{\partial}{\partial y^A_\mu},\frac{\partial}{\partial p}\right\rangle \, , \\
\ker\d\Theta_{\mathcal{W}}=&\left\langle\frac{\partial}{\partial y^A_\mu},\frac{\partial}{\partial s^\mu}\right\rangle \, , \\
\ker\omega=&\left\langle\frac{\partial}{\partial y^A},\frac{\partial}{\partial y^A_\mu},\frac{\partial}{\partial p_A^\mu},\frac{\partial}{\partial p},\frac{\partial}{\partial s^\mu}\right\rangle \, , \\
\mathfrak{R}=&\left\langle\frac{\partial}{\partial p},\frac{\partial}{\partial s^\mu}
\right\rangle \, , \qquad \text{ (since ${\cal A}^m(\ker\omega)=\left<\d^mx\right>$)} \, .
\end{align}
Notice that, according to Definition \ref{multicont}, $\Theta_{\mathcal{W}}$ is a premulticontact form in $\mathcal{W}$, since $\ker\Theta_{\mathcal{W}}\cap\ker\d\Theta_{\mathcal{W}}\neq\{ 0\}$ and 
$\ker\d\Theta_{\mathcal{W}}\not=\{ 0\}$. 
However, the pair $(\Theta_{\mathcal{W}},\omega)$ does not define a special (pre)multicontact structure in $\mathcal{W}$, since
$$
\ker \Theta_\mathcal{W}\cap\ker\d\Theta_\mathcal{W}\cap\ker\omega =
\left\langle\frac{\partial}{\partial y^A_\mu}\right\rangle \, ;
$$
so, $\rank(\ker\Theta_\mathcal{W}\cap\ker\d\Theta_\mathcal{W}\cap\ker\omega)=mn \neq \rank\mathfrak{R}-m=1$ (except when $m=n=1$),
and then conditions (2) and (3) of Definition \ref{multicontactdef} are, in general, incompatible.
\end{remark}

\begin{definition}
Given a Lagrangian density $\mathcal{L}\in\df^m(\mathcal{P})$,
we denote also $\mathcal{L} = \pr^*\mathcal{L}\in \Omega^m(\mathcal{W})$.
Then, the {\sl\textbf{Hamiltonian submanifold}}
 ${\rm j}_\circ\colon \mathcal{W}_\circ \hookrightarrow \mathcal{W}$ is defined by
\[\mathcal{W}_\circ = \{ w \in \mathcal{W} \mid \mathcal{L}(w) = \mathcal{C}(w)\} \, ,\]
\end{definition}

Since $\mathcal{L}$ is a $\tau$-semibasic form in $\mathcal{P}$, then there is a function $L\in\Cinfty(\mathcal{P})$ such that $\mathcal{L}=L\,(\tau^*\omega_M)\in\df^m(\mathcal{P})\equiv L\,\omega$, and its pull-back $\pr^*L\in\Cinfty(\mathcal{W})$ is also denoted by $L\in\Cinfty(\mathcal{W})$; so $\mathcal{L}=L\,((\pr\circ\tau)^*\omega_M)\equiv L\,\omega\in\df^m(\mathcal{W})$. 
From this we obtain the local description of the constraint function defining $\mathcal{W}_\circ$, that is,
    \[ C-L = p + p^\mu_Ay^A_\mu - L(x^\nu, y^B, y^B_\nu, s^\mu) = 0 \, .\]

The local coordinates in $\mathcal{W}_\circ$ are $(x^\mu, y^A, y^A_\mu, p^\mu_A, s^\mu)$, and then we have ${\rm j}_\circ(x^\mu, y^A, y^A_\mu, p^\mu_A, s^\mu) = (x^\mu, y^A, y^A_\mu, p^\mu_A, L - y^A_\mu p^\mu_A, s^\mu)$.

\begin{proposition}
\label{diffeom}
    $\mathcal{W}_\circ$ is a $1$-codimensional $\mu_\mathcal{W}$-transverse submanifold of $\mathcal{W}$, and is diffeomorphic to $\mathcal{W}_r$ (this diffeomorphism is $\mu_\mathcal{W}\circ{\rm j}_\circ\colon\mathcal{W}_\circ\to\mathcal{W}_r$).

\end{proposition}
\begin{proof}
    For every $(x^\mu, y^A, y^A_\mu, p^\mu_A, s^\mu) = (\bar{y}, \mathbf{p}, s) \in \mathcal{W}_\circ$, we have $L(\bar{y}, s) = L(\bar{y}, \mathbf{p}, s) = C(\bar{y}, \mathbf{p}, s)$. Moreover,
    \[
    (\mu_\mathcal{W} \circ {\rm j}_\circ) (\bar{y}, \mathbf{p}, s) = \mu_\mathcal{W}(\bar{y}, \mathbf{p}, s) = (\bar{y}, [\mathbf{p}], s) \, .
    \]
    
    First, $\mu_\mathcal{W}\circ{\rm j}_\circ$ is injective.
    In fact, let $(\bar{y}_1,\mathbf{p}_1,s_1) \, , (\bar{y}_2, \mathbf{p}_2,s_2)\in\mathcal{W}_\circ$, then,
    \[ 
    (\mu_\mathcal{W} \circ {\rm j}_\circ)(\bar{y}_1, \mathbf{p}_1, s_1) = (\mu_\mathcal{W} \circ {\rm j}_\circ)(\bar{y}_2, \mathbf{p}_2,s_2)\Rightarrow(\bar{y}_1,[\mathbf{p}_1],s_1) = (\bar{y}_2, [\mathbf{p}_2], s_2) \, .
    \] 
    Hence $\bar{y}_1=\bar{y}_2$, $[\mathbf{p}_1]=[\mathbf{p}_2]$, $s_1=s_2$, and
    $$
    L(\bar{y}_1,s_1)=L(\bar{y}_2,s_2)=C(\bar{y}_2,\mathbf{p}_2,s_2)= C(\bar{y}_1,\mathbf{p}_1,s_1) \, .
    $$
    In particular, in local coordinates, the equality 
    $C(\bar{y}_2,\mathbf{p}_2,s_2)=C(\bar{y}_1,\mathbf{p}_1,s_1)$ gives
    \[ 
    p(\mathbf{p_2}) + p_A^\mu(\mathbf{p}_2)\, y^A_\mu(\bar{y}_2) = p(\mathbf{p_1}) + p_A^\mu(\mathbf{p}_1)\, y^A_\mu(\bar{y}_1) \, .
    \]
    As $[\mathbf{p}_1] = [\mathbf{p}_2]$, it follows that $p^\mu_A(\mathbf{p}_1) = p^\mu_A([\mathbf{p}_1]) = p^\mu_A([\mathbf{p}_2]) = p^\mu_A(\mathbf{p}_2)$, then also $p(\mathbf{p}_1) = p(\mathbf{p}_2)$. 
    This proves the injectivity.

    Second, $\mu_\mathcal{W}\circ{\rm j}_\circ$ is surjective. 
    In fact, let $(\bar{y},[\mathbf{p}],s)\in\mathcal{W}_r$, then there exists a $(\bar{y},\mathbf{q},s)\in{\rm j}_\circ(\mathcal{W}_\circ)$ such that
    \[p^\mu_A(\mathbf{q})=p^\mu_A([\mathbf{p}]) \text{ and } p(\mathbf{q})=p^\mu_A([\mathbf{p}]) \, y^A_\mu(\bar{y}) - L(\bar{y}) \, ,\]
    that gives $[\mathbf{q}]=[\mathbf{p}]$, since $p^\mu_A([\mathbf{p}])=p^\mu_A(\mathbf{p})$.
    This proves the surjectivity.

    Finally, $\mu_\mathcal{W}$-transversality and $1$-codimension as a submanifold of $\mathcal{W}$, follows from the fact that
    $\displaystyle\ker (\mu_\mathcal{W})_*= \left<\frac{\partial}{\partial p}\right>$ 
    and then, for the constraint defining $\mathcal{W}_\circ$ in $\mathcal{W}$, 
    we have that $\displaystyle\frac{\partial}{\partial p}(L- C)= 1$.
\end{proof}

Now, to define a premulticontact structure in $\mathcal{W}_\circ$, we consider the form 
$$
\Theta_\circ := {\rm j}_\circ^*\Theta_\mathcal{W} \in \Omega^m(\mathcal{W}_\circ)\, ,
$$
whose local expression is:
\begin{equation} 
\label{Theta0}
\Theta_\circ = -p_A^\mu \, \mathrm{d}y^A \wedge \mathrm{d} ^{m-1}x_\mu-(L-y^A_\mu p^\mu_A) \,  \mathrm{d}^m x + \mathrm{d}s^\mu \wedge \mathrm{d}^{m-1}x_\mu \, ,
\end{equation} 
and, for which, 
$$
\d \Theta_\circ =  -\d p_A^\mu \wedge \mathrm{d}y^A \wedge \mathrm{d} ^{m-1}x_\mu - \left( \frac{\partial L}{\partial y^A} \d y^A + \Big(\frac{\partial L}{\partial y^A_\mu}-p^\mu_A\Big) \d y^A_\mu +  \frac{\partial L}{\partial s^\mu} \d s^\mu - y^A_\mu \d p^\mu_A\right) \wedge  \mathrm{d}^m x \, . 
$$
Now, taking a generic vector field $\displaystyle X= f^\mu \frac{\partial}{\partial x^\mu} + f^A \frac{\partial}{\partial y^A} + f^A_\mu \frac{\partial}{\partial y^A_\mu} + g^\mu_A \frac{\partial}{\partial p^\mu_A} +g^\mu \frac{\partial}{\partial s^\mu} \in \mathfrak{X}(\mathcal{W}_\circ)$, we can compute
\begin{align*}
\inn{X}\Theta_\circ = \, & g^\mu\, \mathrm{d}^{m-1}x_\mu - f^A p_A^\mu\,\mathrm{d}^{m-1}x_\mu + \\
 & f^\nu \left[p_A^\mu \, \mathrm{d}y^A \wedge \mathrm{d} ^{m-2}x_{\nu\mu} - (L - y^A_\mu \, p^\mu_A)\, \mathrm{d}^{m-1} x_\nu- \d s^\mu \wedge \d^{m-2} x_{\mu \nu} \right]\, , \\
\inn{X}\d\Theta_\circ = \, & -g^\mu \frac{\partial L}{\partial s^\mu} \, \mathrm{d}^m x - g^\mu_A\left(\mathrm{d}y^A \wedge \mathrm{d}^{m-1}x_\mu-y^A_\mu\mathrm{d}^m x\right) - f^A_\mu \left(\frac{\partial L}{\partial y^A_\mu} -  p^\mu_A\right) \, \mathrm{d}^m x + \\
 & f^A \left(\d p_A^\mu\wedge\mathrm{d}^{m-1}x_\mu-\frac{\partial L}{\partial y^A} \, \mathrm{d}^m x \right)-f^\nu\,\d p_A^\mu \wedge \mathrm{d}y^A \wedge \mathrm{d}^{m-2}x_{\nu \mu}+ \\
 &  f^\nu \left( \frac{\partial L}{\partial y^A} \d y^A + \frac{\partial L}{\partial y^A_\mu} \d y^A_\mu + \frac{\partial L}{\partial s^\mu} \d s^\mu - y^A_\mu\d p^\mu_A-p^\mu_A \d  y^A_\mu \right) \wedge  \mathrm{d}^{m-1} x_\nu \, , 
\end{align*}
from which one obtains
\begin{align*}
\ker\Theta_{\circ} &=
\left\{\displaystyle f^A\bigg(\frac{\partial}{\partial y^A}+p^\mu_A\parder{}{s^\mu}\bigg) + f^A_\mu\frac{\partial}{\partial y^A_\mu} + g^\mu_A\frac{\partial}{\partial p_A^\nu}\right\} \ , \\
\ker\d\Theta_{\circ} &= \left\{ f^A_\mu \frac{\partial}{\partial y^A_\mu}+g^\mu \parder{}{s^\mu} \ \Big\vert\  -g^\mu \frac{\partial L}{\partial s^\mu}-f^A_\mu \left(\frac{\partial L}{\partial y^A_\mu} - p^\mu_A\right) = 0 \right\} \ , \\
\ker\,\omega &= \left\{f^A\frac{\partial}{\partial y^A} + f^A_\mu\frac{\partial}{\partial y^A_\mu}+ g^\mu_A\frac{\partial}{\partial p_A^\mu}+ g^\mu\parder{}{s^\mu} \right\} \, .
\end{align*}
and, using Definition \ref{def:reeb_dist} and bearing in mind that ${\cal A}^m(\ker\omega)=\left<\d^mx\right>$, 
the Reeb distribution is given by
$$
\mathfrak{R}=\left\{f^A_\mu\left(\frac{\partial L}{\partial y^A_\mu}-p^\mu_A\right)\derpar{}{y^A_\mu}+g^\mu\derpar{}{s^\mu}\right\} \, . 
$$
Let us remark that dimension-wise, we have
\begin{align*}
    \rank{\ker \Theta_\circ} &= n + 2nm \,, &
    \rank{\ker\d\Theta_\circ} &= m + nm -1 \, , \\
    \rank{\ker \omega} &= m+n+2nm \,, & 
    \rank\mathfrak{R} &= m + nm \, .
\end{align*}
Therefore, although, in general, $\Theta_\circ$ is a premulticontact form (according to Definition \ref{multicont}), the couple $(\Theta_\circ,\omega)$ does not define a special (pre)multicontact structure in $\mathcal{W}_\circ$.
In fact; we have that $\ker\Theta_\circ\cap\ker\d\Theta_\circ=\{ 0\}$ and 
$\ker\d\Theta_\circ\not=\{ 0\}$, and the conditions in Definition \ref{multicont} hold.
But, as $\ker\Theta_\circ\cap\ker\d\Theta_\circ\cap \ker\,\omega=\{ 0\}$, then
$\rank (\ker\Theta_\circ\cap\ker\d\Theta_\circ \cap\ker\omega)=0\neq\rank\mathfrak{R}-m=nm$, and conditions (2) and (3) in Definition \ref{multicontactdef} are incompatible, unless some additional conditions hold.
Thus, we have the following:

\begin{proposition}
The triad $(\mathcal{W}_\circ,\Theta_\circ,\omega)$ is a special premulticontact bundle on the points of the $\rho_M$-transversal submanifold ${\rm j}_1\colon\mathcal{W}_1\hookrightarrow\mathcal{W}_\circ$ defined as
\begin{align}
\mathcal{W}_1 &:=\{ w\in\mathcal{W}_\circ\mid
(\inn{X}\d\Theta_\circ)(w)=0\quad \text{ for every }\ X\in\vf^{{\rm V}(\overline{\pr}_\circ)}(\mathcal{W}_\circ)\} \label{W1a} \\
&=
\left\{ (x^\mu, y^A, y^A_\mu, p^\mu_A, s^\mu)\in \mathcal{W}_\circ \  \big\vert \  \frac{\partial L}{\partial y^A_\mu} = p^\mu_A \right\} \, .
\label{W1b}
\end{align}
\end{proposition}
\begin{proof}
Thus, bearing in mind these results,
for Definition \ref{multicontactdef} to hold, we have to impose:
\[ \xi^\mu_A\equiv\frac{\partial L}{\partial y^A_\mu} - p^\mu_A = 0 \, .
\label{contraints1} \]
In this case, we have that $\rank{\ker \d\Theta_\circ}=nm+\ell$, where $0\leq \ell\leq m-1$
(this value depends on if $\displaystyle \derpar{L}{s^\mu}=0$, for some $\mu$); then 
$$
\mathfrak{R}=\left\langle \derpar{}{s^\mu}\right\rangle \, , \qquad 
\ker \Theta_\circ \cap \ker \d \Theta_\circ \cap \ker \omega =
\left\langle\frac{\partial}{\partial y^A_\mu}\right\rangle \, ,
$$
and conditions (2) and (3) of Definition \ref{multicontactdef} hold for $k=nm$.
These constraints $\xi^\mu_A$ locally define the submanifold \eqref{W1b} which, as
$\displaystyle \vf^{{\rm V}(\overline{\pr})}(\mathcal{W}_\circ)=\left\langle\frac{\partial}{\partial y^A_\mu}\right\rangle$,
can be characterized as indicated in \eqref{W1a}.
Thus, it is obvious that $\mathcal{W}_1$ is a 
$\rho_M$-transversal submanifold of $\mathcal{W}_\circ$.
\end{proof}

So, we have the following diagram:
\begin{equation}
\begin{tikzcd}[column sep=large, row sep=large]
    & \mathcal{W}
    \arrow[dddl, swap, "\pr", bend right]
    \arrow[dddr, "\overline{\pr}", bend left]
    \\
     & \mathcal{W}_\circ
    \arrow[u, hook, "\mathrm{j}_\circ"]
    \arrow[ddd, "\rho^\circ_M", crossing over, bend right = 35]
    \arrow[ddl, "\pr_\circ", swap, bend right=10]
    \arrow[ddr, "\overline{\pr}_\circ", bend left=10]
    \\
    & \mathcal{W}_1
    \arrow[u, hook, "\mathrm{j}_1"]
    \arrow[dd, "\rho^{1}_M"]
    \arrow[dl, "\pr_1", crossing over]
    \arrow[dr, "\overline{\pr}_1", swap]
    \\
    \mathcal{P} = J^1\pi \times_M \bigwedge\nolimits^{m-1}\cT M
    \arrow[dr, "\tau", swap]
    & & J^{1*}\pi \times_M \bigwedge\nolimits^{m-1} \mathrm{T}^*M=\mathcal{P}^*
    \arrow[dl, "\bar\tau"]
    \\
    & M
\end{tikzcd}
\end{equation}
\noindent And, as is evident, we also have the following:

\begin{corollary}
$\mathcal{W}_1={\rm graph}\,\mathcal{FL}$.
\end{corollary}

Now, we must define the associated dissipation form
$\sigma_{\Theta_\circ}\in\df^1(\mathcal{W}_\circ)$.
Thus, for a generic vector field
$\displaystyle X=f^A_\mu \frac{\partial}{\partial y^A_\mu}+g^\mu \frac{\partial}{\partial s^\mu}\in\mathfrak{R}$,
and taking 
$\sigma_{\Theta_\circ}=\sigma_\mu\d x^\mu+\sigma_A\d y^A+\sigma_A^\mu\d y^A_\mu+\sigma_\mu^A\d p^\mu_A+\alpha_\mu\d s^\mu$
the  \eqref{sigma} leads to
$\sigma_A=\sigma_A^\mu=\sigma_\mu^A=\alpha_\mu=0$, and
$$
0=\Big(\sigma^\mu-\frac{\partial L}{\partial s^\mu}\Big)f^\mu+
\Big( \frac{\partial L}{\partial y^A_\mu} - p^\mu_A\Big)f_\mu^A
\underset{\mathcal{W}_1}{=} \Big(\sigma^\mu-\frac{\partial L}{\partial s^\mu}\Big)f^\mu
\quad \Longrightarrow\quad \sigma_\mu \underset{\mathcal{W}_1}{=}\frac{\partial L}{\partial s^\mu} \ ,
$$
and hence we obtain 
\[
\sigma_{\Theta_\circ} = \frac{\partial L}{\partial s^\mu} \mathrm{d}x^\mu \qquad \text{(on $\mathcal{W}_1$)} \, .
\]

Finally, on the points of $\mathcal{W}_1$, we can define the form
$$
\bd\Theta_\circ = 
\mathrm{d} \Theta_\circ + \sigma_{\Theta_\circ} \wedge \Theta_\circ\in \Omega^m(\mathcal{W}_\circ)\, ,
$$
whose local expression is:
\begin{align}
\bd\Theta_\circ= & - \d p_A^\mu \wedge \d y^A \wedge \d  ^{m-1}x_\mu - \Big( \frac{\partial L}{\partial y^A} - \frac{\partial L}{\partial s^\mu} \, p_A^\mu \Big) \d y^A \wedge \d ^mx \\ & 
- \Big(\frac{\partial L}{\partial y^A_\mu} - p^\mu_A \Big) \d y^A_\mu \wedge \d^mx + y^A_\mu \,\d p^\mu_A \wedge \d^mx 
\label{Omega0a}\\
\underset{\mathcal{W}_1}{=} & - \d p_A^\mu \wedge \d y^A \wedge \d^{m-1}x_\mu - \Big( \frac{\partial L}{\partial y^A} - \frac{\partial L}{\partial s^\mu} \, \frac{\partial L}{\partial y^A_\mu} \Big) \d y^A \wedge \d^mx + y^A_\mu\, \d p_A^\mu \wedge \d^mx
    \, .
    \label{Omega0b}
\end{align}

\begin{remark}
\label{remarkHam}
The result in Proposition \ref{diffeom} allows us to define a Hamiltonian section of the projection $\mu_\mathcal{W}\colon \mathcal{W}\to\mathcal{W}_r$, namely ${\bf h}_{_\mathcal{W}}\colon\mathcal{W}_r\to\mathcal{W}$, such that $\mathcal{W}_\circ={\rm Im}\,{\bf h}_{_\mathcal{W}}$; which, in coordinates, is 
$${\bf h}_{_\mathcal{W}}(x^\mu,y^A,y^A_\mu,p_A^\mu,s^\mu)=(x^\mu,y^A,y^A_\mu,p_A^\mu,p=-{\rm H}(x^\nu,y^B,y^B_\nu,p^\nu_B,s^\nu),s^\mu)\,,$$
and where ${\rm H}\in\Cinfty(\mathcal{W}_r)$ is its associated Hamiltonian function.
We denote also by ${\rm H}$ the function
$(\mu_\mathcal{W}\circ{\rm j}_\circ)^*{\rm H}\in\Cinfty(\mathcal{W}_\circ)$,
since both functions have the same coordinate expression:
$$
{\rm H}(x^\nu,y^B,y^B_\nu,p^\nu_B,s^\nu)=p^\mu_Ay^A_\mu-L(x^\nu,y^B,y^B_\nu,s^\nu)\,.
$$
Using this Hamiltonian section one can construct the premulticontact form 
$\Theta_r:={\bf h}_{_\mathcal{W}}^*\Theta_\mathcal{W} \in \Omega^m(\mathcal{W}_r)$,
which verifies that $\Theta_\circ=(\mu_\mathcal{W}\circ{\rm j}_\circ)^*\Theta_r$ and
whose coordinate expression is also like in \eqref{Theta0}. 
Therefore, all the formalism can also be constructed
starting from the multicontact manifold $(\mathcal{W}_r,\Theta_r)$ 
and introducing the corresponding submanifold $(\mu_\mathcal{W}\circ{\rm j}_\circ)(\mathcal{W}_1)$
of $(\mathcal{W}_r,\Theta_r)$
on which this structure is a special premulticontact one.

Furthermore, from the Hamiltonian section ${\bf h}_{\mathcal{W}}$, it is possible to recover a Hamiltonian section in the multicontact Hamiltonian formalism,  $\widetilde{\mathbf{h}}_\circ\colon\mathcal{P}_\circ^*\rightarrow\widetilde{P}$
(see Diagram \eqref{diaghamsec0}). 
As shown in the following diagram, it is defined by $\widetilde{\mathbf{h}}_\circ([\mathbf{p}],s)=(\widetilde{\pr}\circ{\bf h}_{_\mathcal{W}})\left((\widehat{\pr})^{-1}({\rm j}_\circ([\mathbf{p}],s))\right)$, for every $([\mathbf{p}],s)\in\mathcal{P}^*_\circ$; 
\begin{equation}
\begin{tikzcd}[column sep=huge, row sep=huge]
    \widetilde{\mathcal{P}}_\circ
    \arrow[r, "\widetilde{\jmath_\circ}"]
    \arrow[d, "\mathfrak{p}_\circ"]
    &
    \widetilde{\mathcal{P}}
    \arrow[d, "\widetilde{\mathfrak{p}}"]
    &
    \mathcal{W}
    \arrow[l, "\widetilde{\pr}", swap]
    \arrow[d, "\mu_{\mathcal{W}}", swap]
    \\
    \mathcal{P}^*_\circ
    \arrow[r, "\jmath_\circ"]
    \arrow[ur, "\widetilde{\mathbf{h}}_\circ", swap]
    \arrow[u, "\mathfrak{p}_\circ^{-1} = \mathbf{h}_\circ", bend left]
    &
    \mathcal{P}^*
    &
    \mathcal{W}_r
    \arrow[l, "\widehat{\pr}"]
    \arrow[u, "\mathbf{h}_{\mathcal{W}}", swap, bend right]
\end{tikzcd}
\end{equation}
Note that, for hyper-regular Lagrangians, we have $\mathcal{P}^*_\circ=\mathcal{P}^*$ and $\widetilde{\mathcal{P}}_\circ= \widetilde{\mathcal{P}}$.
\end{remark}

\subsection{Field equations}

As in the Lagrangian and Hamiltonian formalisms,
for the premulticontact Hamiltonian system $({\cal W}_\circ,\Theta_\circ,\omega)$, 
the field equations can be stated in three equivalent ways (at least locally). 
Thus, the {\sl\textbf{premulticontact Lagrange--Hamilton problem}} in $\mathcal{W}_\circ$ consists in finding sections $\bm\psi_\circ\colon M \rightarrow \mathcal{W}_\circ$, with image in $\mathcal{W}_1$,
such that:
\begin{itemize}
\item [(1)]
They are the solutions to the {\sl\textbf{premulticontact equations for sections}}:
\begin{equation} 
\label{Hamilton-Lagrange eq}
\bm\psi_\circ^*\Theta_\circ\underset{\mathcal{W}_1}{=} 0\,, \quad \bm\psi_\circ^* \ i(Y_\circ)\bd\Theta_\circ \underset{\mathcal{W}_1}{=} 0\,, \qquad \text{ for every $Y_\circ \in \mathfrak{X}(\mathcal{W}_\circ)$}\,, 
\end{equation}
or, equivalently, for their canonical prolongations,
\begin{equation}
\label{Hamilton-Lagrange eq2}
\inn{\bm{\psi}^{(m)}}(\Theta_{\circ}\circ\bm{\psi})\underset{\mathcal{W}_1}{=}0 \,,\qquad
\inn{\bm{\psi}^{(m)}}(\bd\Theta_{\circ}\circ\bm{\psi}) \underset{\mathcal{W}_1}{=} 0\,.
\end{equation}
\item [(2)]
They are the integral sections of classes of integrable and $\rho^\circ_M$-transverse $m$-multivector fields $\{X_\circ\} \subset \mathfrak{X}^m(\mathcal{W}_\circ)$, tangent to  $\mathcal{W}_1$, which are solutions to the 
{\sl\textbf{premulticontact equations for multivector fields}}:
\begin{equation}
\label{vfH20}
\inn{{\bf X}_\circ}\Theta_\circ \underset{\mathcal{W}_1}{=} 0 \, , \quad \inn{\bfX_\circ}\bd\Theta_\circ\underset{\mathcal{W}_1}{=} 0 \, , \qquad \text{ for every $\mathbf{X}_\circ \in \{\mathbf{X}_\circ$}\} \, .
\end{equation}
\item [(3)]  
They are the integral sections of integrable Ehresmann connections $\nabla_\circ\in\mathcal{T}^1_1(\mathcal{W}_\circ)$, whose horizontal distributions are tangent to  $\mathcal{W}_1$, which are solutions to the {\sl\textbf{premulticontact equations for Ehresmann connections}}
\begin{equation}
\label{EcH0}
\inn{\nabla_\circ}\Theta_\circ\underset{\mathcal{W}_1}{=}(m-1)\Theta_\circ \,,\qquad
\inn{\nabla_\circ}\bd\Theta_\circ\underset{\mathcal{W}_1}{=}(m-1)\bd\Theta_\circ \,.
\end{equation}
\end{itemize}
Locally decomposable and $\rho^\circ_M$-transverse multivector fields and orientable connections, which are solutions to equations
\eqref{vfH20} and \eqref{EcH0}, respectively, are called {\sl\textbf{Lagrange--Hamiltonian multivector fields}} and {\sl\textbf{connections}} for the premulticontact system $(\mathcal{W}_\circ, \Theta_\circ,\omega)$. \textsl{Euler--Lagrange} and \textsl{Hamilton--de Donder--Weyl} multivector fields and connections can be recovered from these Lagrange--Hamiltonian multivector fields and connections, as we will see later.

\subsubsection*{Field equations for multivector fields and Ehresmann connections}

In order to analyze the information contained in the equations,
let us consider the expression in a natural chart of coordinates of ${\cal W}_\circ$ of
a $\rho^\circ_M$-transverse and locally decomposable $m$-multivector field
satisfying $\inn{\bfX_\circ}\omega=1$,
$$
{\bf X}_\circ= \bigwedge_{\mu=1}^m
\bigg(\parder{}{x^\mu}+(X_\circ)_\mu^A\frac{\displaystyle\partial}{\displaystyle
\partial y^A}+(X_\circ)_{\mu\nu}^A\frac{\displaystyle\partial}{\displaystyle\partial y^A_\nu}+
(X_\circ)_{\mu A}^\nu\frac{\displaystyle\partial}{\displaystyle\partial p_A^\nu}+(X_\circ)_\mu^\nu\,\frac{\partial}{\partial s^\nu}\bigg) \,;
$$
and of the Ehresmann connection $\nabla_\circ$
associated with the class $\{{\bf X}_\circ\}$ is,
$$\displaystyle
\nabla_\circ= \d x^\mu\otimes
\bigg(\parder{}{x^\mu}+(X_\circ)_\mu^A\frac{\displaystyle\partial}{\displaystyle
\partial y^A}+(X_\circ)_{\mu\nu}^A\frac{\displaystyle\partial}{\displaystyle\partial y^A_\nu}+
(X_\circ)_{\mu A}^\nu\frac{\displaystyle\partial}{\displaystyle\partial p_A^\nu}+(X_\circ)_\mu^\nu\,\frac{\partial}{\partial s^\nu}\bigg) \,,
$$
Then, bearing in mind the local expressions \eqref{Theta0} and \eqref{Omega0a},
the second equation of \eqref{vfH20} or \eqref{EcH0} gives
\begin{align}
(X_\circ)_\mu^A =y^A_\mu  \qquad \text{ (on $\mathcal{W}_1$)} \,, \label{coefs0a} \\
(X_\circ)_{\mu A}^\mu =
\frac{\partial L}{\partial y^A}+\frac{\partial L}{\partial s^\mu}\, p_A^\mu  \qquad \text{ (on $\mathcal{W}_1$)} \,, 
\label{coefs0b} \\
\frac{\partial L}{\partial y^A_\mu} - p^\mu_A = 0
 \qquad \text{ (on $\mathcal{W}_1$)} \, ,
\label{coefs0c}
\end{align}
and another set of equalities which are identities when the above equations are taken into account. 
Now, using \eqref{coefs0a}, the first equation of \eqref{vfH20} or \eqref{EcH0} leads to
\begin{equation}
\label{firsteqs}
(X_\circ)_\mu^\mu =L \qquad \text{ (on $\mathcal{W}_1$)} \,.
\end{equation}
Equations \eqref{coefs0c} are compatibility constraints for the field equations \eqref{vfH20} and \eqref{EcH0},
and are satisfied identically, since they are also the constraints \eqref{contraints1} defining the submanifold $\mathcal{W}_1$ in $\mathcal{W}_\circ$, and the field equations hold at support on this submanifold.
The other equations \eqref{coefs0a} and \eqref{coefs0b} partially determine the multivector fields and connections solutions to 
the equations \eqref{vfH20} or \eqref{EcH0},
which are
\begin{align}
{\bf X}_\circ & \underset{\mathcal{W}_1}{=} \,  \bigwedge_{\mu=1}^m
\bigg(\parder{}{x^\mu}+y_\mu^A\frac{\partial}{\partial y^A}+(X_\circ)_{\mu\lambda}^A\frac{\partial}{\partial y^A_\lambda}+
(X_\circ)_{\mu A}^\lambda\frac{\partial}{\partial p_A^\lambda}+(X_\circ)_\mu^\lambda\,\frac{\partial}{\partial s^\lambda}\bigg)\equiv 
\bigwedge_{\mu=1}^m X_{\circ\,\mu}\,,  \\
\nabla_\circ & \underset{\mathcal{W}_1}{=} \, \d x^\mu\otimes
\bigg(\parder{}{x^\mu}+y_\mu^A\frac{\partial}{\partial y^A}+(X_\circ)_{\mu\lambda}^A\frac{\partial}{\partial y^A_\lambda}+
(X_\circ)_{\mu A}^\lambda\frac{\partial}{\partial p_A^\lambda}+(X_\circ)_\mu^\lambda\,\frac{\partial}{\partial s^\lambda}\bigg)\equiv
\d x^\mu\otimes X_{\circ\,\mu}  \,,
\end{align}
where the components $(X_\circ)_{\mu A}^\mu$ are related
by the equations \eqref{coefs0b}.
In particular, equations \eqref{coefs0a} set the semi-holonomy condition of multivector fields and connections solution.

Next, we have to impose the condition that these solutions are tangent to the submanifold $\mathcal{W}_1$, which gives rise to the following equations:
\begin{align}
0 & \underset{\mathcal{W}_1}{=}
\Lie_{X_{\circ\,\mu}}\xi^\nu_B =
\Lie_{X_{\circ\,\mu}}\bigg(\frac{\partial L}{\partial y^B_\nu} - p^\nu_B\bigg)
\nonumber \\ & =
\parderr{L}{x^\mu}{y_\nu^B}
+\frac{\partial^2L}{\partial y^A \partial y^B_\nu}\,y_\mu^A
+\frac{\partial^2L}{\partial y^A_\lambda\partial y^B_\nu}(X_\circ)_{\mu\lambda}^A
-(X_\circ)_{\mu B}^\nu
+\frac{\partial^2L}{\partial s^\lambda\partial y^B_\nu}(X_\circ)_\mu^\lambda \,. 
\label{tangeqs}
\end{align}
This is a system of linear equations for the coefficients of the multivector field $\mathbf{X}_\circ$ or the connection $\nabla_\circ$.
In particular, we can resolve the coefficients $(X_\circ)_{\mu B}^\nu$
as functions of the arbitrary coefficients $(X_\circ)_{\mu\lambda}^A$ and the coefficients $(X_\circ)_\mu^\lambda$. 
But the result must be compatible with equations \eqref{coefs0b};
so, taking the coefficients $(X_\circ)_{\mu B}^\mu$ in the last equation and combining with  \eqref{coefs0b}, we get
\begin{equation}
\frac{\partial L}{\partial y^B}+\frac{\partial L}{\partial s^\mu}\, p_B^\mu \, \underset{\mathcal{W}_1}{=} \,
\parderr{L}{x^\mu}{y_\mu^B} 
+\frac{\partial^2L}{\partial y^A \partial y^B_\mu}y_\mu^A 
+\frac{\partial^2L}{\partial y^A_\lambda\partial y^B_\mu}(X_\circ)_{\mu\lambda}^A
+\frac{\partial^2L}{\partial s^\lambda\partial y^B_\mu}(X_\circ)_\mu^\lambda \, ,
\label{unieqs}
\end{equation}
which is a new system of linear equations for the coefficients
$(X_\circ)_{\mu\lambda}^A$ and $(X_\circ)_\mu^\lambda$.
Then, we have two possibilities:
\begin{itemize}
\item 
If the Lagrangian $L$ is regular, then the Hessian matrix $\displaystyle\bigg(\frac{\partial^2L}{\partial y^A_\lambda\partial y^B_\mu}\bigg)$ is regular everywhere and in the system \eqref{unieqs}, this allows us to isolate the coefficients $(X_\circ)_{\mu\lambda}^A$,
although they cannot be completely determined
as functions of the remaining coefficients 
$(X_\circ)_\mu^\nu$, unless $m=1$.
\item 
If the Lagrangian $L$ is singular, 
the Hessian matrix $\displaystyle\bigg(\frac{\partial^2L}{\partial y^A_\lambda\partial y^B_\mu}\bigg)$ is singular too
and the system \eqref{unieqs} could be incompatible,
unless some additional constraints are satisfied.
These constraints define locally a new submanifold
$\mathcal{W}_2\hookrightarrow\mathcal{W}_1$
where solutions for the coefficients $(X_\circ)_{\mu\lambda}^A$ exist.
In this case, the constraint algorithm continues by demanding the tangency condition for the new multivector fields or connections found,
until finding a final constraint submanifold ${\cal W}_f\hookrightarrow\mathcal{W}_1$, where tangent solutions exist (if there are).
\end{itemize}
Observe that, in any case, the semi-holonomy condition is always ensured by equations \eqref{coefs0a}.

Finally, for the integral sections of $\bfX_\circ$
and $\nabla_\circ$,
\begin{equation}
\bm{\psi}_\circ(x^\nu)=\left(x^\mu,y^A(x^\nu),y^A_\mu(x^\nu),p^\mu_A(x^\nu),s^\mu(x^\nu)\right) \, ,
\label{intsec}
\end{equation}
we have that
$$
(X_\circ)_\mu^A= \derpar{y^A}{x^\mu} \,, \qquad 
(X_\circ)_{\nu\mu}^A=\derpar{y^A_\mu}{x^\nu}\,, \qquad 
(X_\circ)_{\mu A}^\nu=\derpar{p_A^\mu}{x^\nu}\,, \qquad 
(X_\circ)_\mu^\nu=\derpar{s^\nu}{x^\mu} \,.
$$
Then, taking into account equations \eqref{coefs0a}
(the semi-holonomy condition), we get
\begin{equation}
y_\mu^A= \derpar{y^A}{x^\mu}\,,
\label{sopdecond}
\end{equation}
which is the holonomy condition for the sections,
that appears naturally within the unified formalism, and it is not necessary to impose it by hand in the singular case, as in the Lagrangian formalism.
Therefore, $\displaystyle (X_\circ)_{\nu\mu}^A=\frac{\partial^2y^A}{\partial x^\nu \partial x^\mu}$,
and equations \eqref{firsteqs} and \eqref{coefs0b} transform into,
\begin{align}
 \parder{s^\mu}{x^\mu} &= L\circ{\bm{\psi}_\circ}  \qquad \text{ (on $\mathcal{W}_1$)}
 \,, \label{finaleqs1} \\
\derpar{p_A^\mu}{x^\mu} &=
\bigg(\frac{\partial L}{\partial y^A}+
\displaystyle\frac{\partial L}{\partial s^\mu} p_A^\mu\bigg)\circ{\bm{\psi}_\circ}
 \qquad \text{ (on $\mathcal{W}_1$)} \, .
\label{finaleqs2}
\end{align}
These are called the {\sl\textbf{Herglotz--Lagrangian--Hamiltonian equations}}.

\subsubsection*{Field equations for sections}

The above equations can be equivalently obtained starting from the field equations for sections \eqref{Hamilton-Lagrange eq}. 
In fact, taking a section like \eqref{intsec} and $\displaystyle\Big\{\parder{}{x^\mu},\derpar{}{y^B},\derpar{}{y^B_\mu},\derpar{}{p_B^\mu},\derpar{}{s^\mu}\Big\}$ as a 
basis for the local vector fields in $\mathcal{W}_\circ$, 
and, bearing in mind the local expression \eqref{Omega0a}, we get
\begin{align}
0 & \underset{\mathcal{W}_1}{=}\, \bm\psi_\circ^*\Theta_\circ =
\Big(-p_A^\mu\derpar{y^A}{x^\mu}-L+y^A_\mu p^\mu_A+\derpar{s^\mu}{x^\mu}\Big)\,\mathrm{d}^m x \, ,  \\
0 & \underset{\mathcal{W}_1}{=}\bm\psi_\circ^*\inn{\tparder{}{p_B^\mu}}\bd\Theta_\circ = \Big(-\derpar{y^A}{x^\mu}+y^A_\mu \Big)\d^mx \, , \label{prvert} \\
0 & \underset{\mathcal{W}_1}{=}\bm\psi_\circ^*\inn{\tparder{}{y^B}}\bd\Theta_\circ = \bigg(\derpar{p^\mu_A}{x^\mu}-\Big( \frac{\partial L}{\partial y^A}-\frac{\partial L}{\partial s^\mu}\,p_A^\mu\Big)\bigg)\d^mx \, , \\
0 & \underset{\mathcal{W}_1}{=}\bm\psi_\circ^*\inn{\tparder{}{y^B_\mu}}\bd\Theta_\circ = \bigg(\frac{\partial L}{\partial y^A_\mu} - p^\mu_A\bigg)\d^mx \, .
\end{align}
The last equations hold identically as a consequence of the constraints \eqref{coefs0c}, and it guarantees that the image of section $\bm\psi_\circ$ is on $\mathcal{W}_1={\rm graph}\,\mathcal{FL}$.
Equations $\displaystyle\bm\psi_\circ^*\inn{\tparder{}{x^\mu}}\bd\Theta_\circ\underset{\mathcal{W}_1}{=} 0$ also hold identically using the other equations.
Finally, equations $\displaystyle\bm\psi_\circ^*\inn{\tparder{}{s^\mu}}\bd\Theta_\circ\underset{\mathcal{W}_1}{=} 0$ 
hold because $\displaystyle\derpar{}{s^\mu}\in\ker\bd\Theta_\circ$.
From here the equalities \eqref{sopdecond}, \eqref{finaleqs1}, and \eqref{finaleqs2} are obtained,
which proves the equivalence between the three types of field equations
\eqref{Hamilton-Lagrange eq}, \eqref{vfH20}, and \eqref{EcH0}.

Furthermore, if \eqref{intsec} is an integral section of every integrable multivector field ${\bf X}_\circ$ solution to \eqref{vfH20}, s \eqref{tangeqs} give
$$
\derpar{p_\nu^A}{x^\mu}\underset{\mathcal{W}_1}{=}
\parderr{L}{x^\mu}{y_\nu^B}
+\frac{\partial^2L}{\partial y^A \partial y^B_\nu}\derpar{y^A}{x^\mu}
+\frac{\partial^2L}{\partial y^A_\lambda\partial y^B_\nu}\derpar{y_\lambda^A}{x^\mu}
+\frac{\partial^2L}{\partial s^\lambda\partial y^B_\nu}\derpar{s^\lambda}{x^\mu}=\derpar{}{x^\mu}\Big(\parder{L}{y_\nu^B}\Big) \, , 
$$
which hold as a consequence of the constraints \eqref{coefs0c};
that is, the Legendre transform,
and hence, they are neither new equations nor new constraints.

\subsection{Relation with the Lagrangian and the Hamiltonian formalisms}

The relation among the Lagrangian, the Hamiltonian, and the unified Lagrangian-Hamiltonian formalisms can be established starting from the field equations for sections \eqref{Hamilton-Lagrange eq}, as follows:

\begin{theorem}
\label{relation}
(See Diagram \eqref{unified} below).
Let $({\cal W}_\circ,\Theta_\circ,\omega)$ be a premulticontact Hamiltonian system described by a (hyper)regular Lagrangian.
If $\bm\psi_\circ\colon M \to \mathcal{W}_\circ$ is a section fulfilling s \eqref{Hamilton-Lagrange eq}
and such that ${\rm Im}\,\bm\psi_\circ\subset\mathcal{W}_1$, then we can construct the sections $\bm\psi_\mathcal{L}=\pr_\circ\circ \,\bm\psi_\circ \colon M\to\mathcal{P}$ and
$\bm\psi_\mathcal{H}=\mathcal{FL}\circ \bm{\psi_\mathcal{L}}=\overline{\pr}_\circ\circ\bm\psi_\circ\colon M\to\mathcal{P}^*$
which verify that:
\begin{enumerate}
\item 
The section $\bm\psi_\mathcal{L}$ is the canonical lift to $\mathcal{P}$ of the projected section $\phi=\rho_E\circ\pr_\circ\circ\,\bm\psi_\circ\colon M \rightarrow E$, and so $\bm\psi_\mathcal{L}$,
is a holonomic section.
\item 
The sections $\bm\psi_\mathcal{L}$ and $\bm\psi_\mathcal{H}$ are solutions to the Lagrangian equations \eqref{sect1H} and the Hamiltonian equations \eqref{sect1H0}, respectively.
\end{enumerate}
Conversely, for every section $\bm\psi_\mathcal{L}\colon M\rightarrow \mathcal{P}$ solution to the Lagrangian equations \eqref{sect1H},
the section $\bm\psi_\circ={\rm j}_1\circ pr_1^{-1}\circ\bm\psi_\mathcal{L}={\rm j}_1\circ\bm\psi_1\colon M \to \mathcal{W}_\circ$ is a solution to the field equations \eqref{Hamilton-Lagrange eq}
and, obviously, ${\rm Im}\,\bm\psi_\circ\subset\mathcal{W}_1$.
\begin{equation}\label{unified}
\begin{tikzcd}[column sep=huge, row sep=huge]
    &
    \mathcal{W}
    \\
    &
    \mathcal{W}_\circ
    \arrow[u, "\mathbf{j}_\circ", swap, hook]
    \arrow[ddl, "\pr_\circ", swap, bend right=20]
    \arrow[ddr, "\overline{\pr}_\circ", bend left=20]
    \\
    &
    \mathcal{W}_1
    \arrow[u, "\mathbf{j}_1", swap, hook]
    \arrow[dl, "\pr_1", swap, pos=0.6]
    \arrow[dr, "\overline{\pr}_1", pos=0.6]
    \arrow[ddd]
    \\
    \mathcal{P} = J^1\pi \times_M \bigwedge\nolimits^{m-1} \cT M
    \arrow[d, "\rho", swap]
    \arrow[ddr, "\rho_E", pos=0.3]
    \arrow[rr, "\mathcal{FL}", crossing over, pos=0.35]
    &&
    J^{1*}\pi \times_M \bigwedge\nolimits^{m-1} \cT M = \mathcal{P}^*
    \arrow[d, "\varrho"]
    \arrow[ddl, "\varrho_E", swap, pos=0.3]
    \\
    J^1\pi
    \arrow[dr, "\pi^1", pos=0.3]
    \arrow[ddr, "\bar\pi^1", swap, bend right=10]
    &&
    J^{1*}\pi
    \arrow[dl, "\kappa^1", swap, pos=0.3]
    \arrow[ddl, "\bar\kappa^1", bend left=10]
    \\
    &
    E
    \arrow[d, "\pi"]
    \\
    &
    M
    \arrow[u, "\phi", pos=0.55, bend left=25, shorten=0.75mm]
    \arrow[uuuuu, "\bm{\psi}_\circ", swap, bend left=40, crossing over]
    \arrow[uuuu, "\bm{\psi}_1", bend right, crossing over, pos=0.66]
    \arrow[uuul, "\bm{\psi}_{\mathcal{L}}", bend left=50, end anchor={[xshift=8ex]south west}]
    \arrow[uuur, "\bm{\psi}_{\mathcal{H}}", swap, bend right=50, end anchor={[xshift=-8ex]south east}]
\end{tikzcd}
\end{equation}
\end{theorem}
\begin{proof}
Let $\bm\psi_\circ\colon M\to\mathcal{W}_\circ$,
be a section solution to the field equations \eqref{Hamilton-Lagrange eq},
whose local coordinate expression is \eqref{intsec}.
As we have seen in the previous section,
the holonomy condition, given by  \eqref{sopdecond}, appears naturally within the unified formalism. 
As $\bm\psi_\circ$ takes value in $\mathcal{W}_1$, 
we can define a section $\bm\psi_1\colon M\to\mathcal{W}_1$ such that
$\bm\psi_\circ={\rm j}_1\circ\bm\psi_1$,
and we have the local expression
$$
\bm{\psi}_\circ(x^\nu)\underset{\mathcal{W}_1}{=}\Big(x^\mu,y^A(x^\nu),\derpar{y}{x^\mu}(x^\nu),\derpar{L}{y_\mu^A}(x^\nu),s^\mu(x^\nu)\Big)
=\bm{\psi}_1(x^\nu) \, .
$$
Then we can construct the sections
$$
\bm\psi_\mathcal{L}:=\pr_\circ\circ \,\bm\psi_\circ=\pr_1\circ \,\bm\psi_1\colon M\to\mathcal{P} \,, \qquad
\bm\psi_\mathcal{H}:=\overline{\pr}_\circ\circ\bm\psi_\circ=\overline{\pr}_1\circ\bm\psi_1\colon M\to\mathcal{P}^* \,,
$$
and, since $\mathcal{W}_1=\graph\,\mathcal{FL}$,
we have $\bm\psi_\mathcal{H}=\mathcal{FL}\circ\bm\psi_\mathcal{L}$.
Their coordinate expressions are:
$$
\bm{\psi}_\mathcal{L}(x^\nu)=\Big(x^\mu,y^A(x^\nu),\derpar{y}{x^\mu}(x^\nu),s^\mu(x^\nu)\Big) \, , \quad
\bm{\psi}_\mathcal{H}(x^\nu)=\Big(x^\mu,y^A(x^\nu),\derpar{L}{y_\mu^A}(x^\nu),s^\mu(x^\nu)\Big) \, ;
$$
and it is customary to write 
$\bm\psi_\circ=(\bm\psi_\mathcal{L}, \mathcal{FL}\circ\bm\psi_\mathcal{L})$.

Notice that, as $\mathcal{W}_1=\graph\,\mathcal{FL}$, then $\mathcal{W}_1$ and $\mathcal{P}$ are diffeomorphic, 
being $\pr_1\colon\mathcal{W}_1\to\mathcal{P}$
(the restriction of the projection $\pr_\circ$ to $\mathcal{W}_1$) the corresponding diffeomorphism.
\begin{enumerate}
\item 
Since $\bm\psi_\circ$ is a holonomic section, so is $\bm\psi_\mathcal{L}$,
as the above coordinate expression shows,
and hence, clearly, $\bm\psi_\mathcal{L}$ is the canonical lift of the projected section $\phi=\rho_E\circ\pr_\circ\circ\,\bm\psi_\circ\colon M \rightarrow E$,
whose coordinate expression is $\phi(x^\nu)=(x^\mu,y^A(x^\nu))$.
\item 
From the corresponding coordinate expressions, it is easy to see that
$\pr_1^*\Theta_\mathcal{L}={\rm j}_1^*\Theta_\circ$.
Then, for every $X\in\vf(\mathcal{P})$, first we have,
\[
\bm\psi_\mathcal{L}^*\Theta_\mathcal{L}=
\bm\psi_\mathcal{L}^*({\rm j}_1\circ \pr_1^{-1})^*\Theta_\circ=
({\rm j}_1\circ \pr_1^{-1}\circ \,\bm\psi_\mathcal{L})^*\Theta_\circ=
\bm\psi_\circ^*\Theta_\circ
\underset{\mathcal{W}_1}{=} 0 \, ,
\label{from psi0 to psiL A}
\]
and, second,
\begin{equation}
\begin{aligned}
\bm\psi_\mathcal{L}^*\inn{X}\bd\Theta_\mathcal{L}=\, & (\pr_\circ\circ \, \bm\psi_\circ)^*\inn{X}\bd\Theta_\mathcal{L}=(\pr_1\circ\,\bm\psi_1)^*\inn{X}\bd\Theta_\mathcal{L}= (\bm\psi_1^*\circ \pr_1^*)\left(\inn{X} \bd\Theta_\mathcal{L}\right)=
\\ = \, & \bm\psi_1^*\left(\inn{(\pr_1^{-1})_*X} \, \pr_1^*\bd\Theta_\mathcal{L}\right)= \bm\psi_1^*\left(\inn{(\pr_1^{-1})_*X} \, {\rm j}_1^*\bd\Theta_\circ\right) \\
\underset{\mathcal{W}_1}{=} \, & (\bm\psi_1^*\circ {\rm j}_1^*) \inn{X_\circ}\bd\Theta_\circ =
\psi_\circ^*\inn{X_\circ}\bd\Theta_\circ \underset{\mathcal{W}_1}{=} 0\, ;
\label{from psi0 to psiL}
\end{aligned}
\end{equation}
since $X_\circ\in\vf(\mathcal{W}_\circ)$ is a vector field tangent to $\mathcal{W}_\circ$,
and the last equality holds for every vector field in $\mathcal{\W}_\circ$, under s \eqref{Hamilton-Lagrange eq}.

The proof for the Hamiltonian section $\bm\psi_\mathcal{H}$
follows the same patterns as for $\bm\psi_\mathcal{L}$.
First,
$$
\bm\psi_\mathcal{H}^*\,\Theta_\mathcal{H}=
\bm\psi_\circ^*\,\overline{\pr}_\circ^*\,\Theta_\mathcal{H}=
\bm\psi_\circ^*\,{\rm j}_1^*\,\overline{\pr}_\circ^*\,\Theta_\mathcal{H}=
\bm\psi_\circ^*\,\Theta_\circ
\underset{\mathcal{W}_1}{=} 0 \, .
$$
The proof of the second equation 
$\displaystyle \bm\psi_\mathcal{H}^*\,\bd\Theta_\mathcal{H}\underset{\mathcal{W}_1}{=} 0$,
is like in \eqref{from psi0 to psiL},
since the projection $\overline{\pr}_1$ is also a diffeomorphism.
\end{enumerate}

Conversely, consider any section $\bm\psi_\mathcal{L}\colon M \to\mathcal{P}$ solution to the s \eqref{sect1H}.
Observe that $\overline{\pr}_1=\mathcal{FL}\circ{\pr}_1$ is a (local) diffeomorphism
and then we can construct the section
$\bm\psi_\circ={\rm j}_1\circ \pr_1^{-1}\circ \, \bm\psi_\mathcal{L}={\rm j}_1\circ\bm\psi_1\colon M\to\mathcal{W}_\circ$.
Then, ${\rm Im}\,\bm\psi_\circ\subset\mathcal{W}_1$ and, for every vector field $X_\circ\in \mathfrak{X}(\mathcal{W}_\circ)$,
following the same reasoning that in \eqref{from psi0 to psiL A}, we obtain that
$\bm\psi_\circ^*\Theta_\circ\underset{\mathcal{W}_1}{=} 0$.
For the second equation,
we split it into $X_\circ=X_\circ^1+X_\circ^2$, with $X_\circ^1\in\vf(\mathcal{W}_\circ)$ tangent to $\mathcal{W}_1$ 
and $X_\circ^2\in\vf^{V(\pr_\circ)}(\mathcal{W}_\circ)$; and hence,
\[ \bm\psi_\circ^*\inn{X_\circ}\bd\Theta_\circ=\bm\psi_\circ^*\inn{X_\circ^1}\bd\Theta_\circ+\bm\psi_\circ^*\inn{X_\circ^2}\bd\Theta_\circ \, .\]
Therefore, on the one hand, following the same reasoning as in \eqref{from psi0 to psiL}, we obtain that
\[ \bm\psi_\circ^*\inn{X_\circ^1}\bd\Theta_\circ=({\rm j}_1\circ pr_1^{-1}\circ \bm\psi_\mathcal{L})^*\inn{X_\circ^1}\bd\Theta_\circ=\bm\psi_\mathcal{L}^*\inn{X^1}\bd\Theta_\mathcal{L}\underset{\mathcal{W}_1}{=} 0 \, , \]
where $X^1=(\pr_1^{-1})_*X_\circ^1\in\vf(\mathcal{P})$.
On the other hand, taking $\displaystyle\Big\{\derpar{}{p^\mu_A}\Big\}$ as a local basis of
$\vf^{V(\pr_\circ)}(\mathcal{W}_\circ)$,
from \eqref{prvert}, we obtain that $\displaystyle\bm\psi_\circ^*\inn{\tparder{}{p^\mu_A}}\bd\Theta_\circ\underset{\mathcal{W}_1}{=} 0$,
since, as $\bm\psi_\mathcal{L}$ must be a holonomic section, the conditions
$\displaystyle y^A_\mu=\derpar{y^ A}{x^\mu}$ hold.
Thus, these results lead to $\bm\psi_\circ^*\inn{X_\circ}\bd\Theta_\circ\underset{\mathcal{W}_1}{=} 0$.

Finally, let $\bm\psi_\mathcal{H}\colon M\rightarrow \mathcal{P}^*$ be a section solution to the Hamiltonian equations \eqref{sect1H0}.
Then, we construct the section $\phi=\varrho_E\circ\bm\psi_\mathcal{H}\colon M\to E$ and the canonical lift
$j^1\bm\phi=(j^1\phi,s)\colon M\to\mathcal{P}$.
Therefore, the equivalence between the Lagrangian and the Hamiltonian formalisms states that $j^1\bm\phi$ is a solution to the Lagrangian equations \eqref{sect1H} and, hence,
from it, a section $\bm\psi_\circ\colon M\to\mathcal{W}_\circ$ solution to the field equations \eqref{Hamilton-Lagrange eq} is obtained following the procedure set out above.
\end{proof}

\begin{remark}
\label{singequiv}
For the case of singular (almost-regular) systems
we have an analogous statement,
but taking ${\cal P}_\circ^*:= {\cal FL}({\cal P})=P_\circ\times\bigwedge\nolimits^{m-1}\Tan^*M$
instead of ${\cal P}^*$,
and considering the premulticontact Hamiltonian system
$(\mathcal{P}_\circ^*,\Theta_\mathcal{H}^\circ,\omega)$
(see the end of Section \ref{multcontHam}).
Then we have the following diagram instead of Diagram \eqref{unified},
\begin{equation}\label{unified2}
\begin{tikzcd}[column sep=huge, row sep=huge]
    &
    \mathcal{W}_\circ
    \arrow[ddl, "\pr_\circ", swap, bend right=20]
    \arrow[dr, "\overline{\pr}_\circ", bend left=10]
    \\
    &
    \mathcal{W}_1
    \arrow[u, "\mathbf{j}_1", swap, hook]
    \arrow[dl, "\pr_1", swap, pos=0.6]
    \arrow[dr, "\overline{\pr}_1^\circ", pos=0.6]
    \arrow[dd]
    &
    J^{1*}\pi \times_M \bigwedge\nolimits^{m-1} \cT M = \mathcal{P}^*
    \\
    \mathcal{P} = J^1\pi \times_M \bigwedge\nolimits^{m-1} \cT M
    \arrow[rr, "\mathcal{FL}_\circ", crossing over, pos=0.12, swap]
    &&
    P_\circ \times_M \bigwedge\nolimits^{m-1} \cT M = \mathcal{P}^*_\circ
    \arrow[u, "\jmath_\circ", swap, hook]
    \\
    &
    M
    \arrow[uuu, "\bm{\psi}_\circ", bend left=30, crossing over, pos=0.46]
    \arrow[uu, "\bm{\psi}_1", bend right, crossing over, pos=0.66, swap, pos=0.65]
    \arrow[ul, "\bm{\psi}_{\mathcal{L}}", bend left=10]
    \arrow[ur, "\bm{\psi}_{\mathcal{H}}^\circ", swap, bend right=10]
\end{tikzcd}
\end{equation}
Then, the proof is the same as for Theorem \ref{relation} but taking into account that, now, the projection $\overline{\pr}_1^\circ$ is not a diffeomorphism but a submersion
(and, in particular, the vector fields $X_\circ\in\vf(\mathcal{W}_\circ)$ appearing in the  \eqref{from psi0 to psiL} must be tangent to $\mathcal{W}_1$ and such that their restrictions to $\mathcal{W}_1$ are $\overline{\pr}_1^\circ$-projectable).
In addition, in the case of singular Lagrangian systems with a
final constraint submanifold ${\cal W}_f\hookrightarrow\mathcal{W}_1$,
the vector fields must be tangent to this submanifold and
all results must hold on the points of ${\cal W}_f$ instead of ${\cal W}_1$.
\end{remark}

\begin{remark}
\label{recovering}
It is easy to obtain the above relations in coordinates.
On the one hand, starting from equations \eqref{finaleqs1} and \eqref{finaleqs2},
as they hold at support on $\mathcal{W}_1$,
using the constraints \eqref{contraints1}
(i.e., the Legendre map) and substituting the multimomenta $p^\mu_A$, these equations transform into
\begin{align}
 \parder{s^\mu}{x^\mu} &= L\circ{\bm{\psi}_{\mathcal{L}}} \,, \\
\frac{\partial}{\partial x^\mu}
\left(\frac{\displaystyle\partial L}{\partial
y^B_\mu}\circ{\bm{\psi}_{\mathcal{L}}}\right) &= 
\left(\frac{\partial L}{\partial y^A}+
\displaystyle\frac{\partial L}{\partial s^\mu}\displaystyle\frac{\partial L}{\partial y^A_\mu}\right)\circ{\bm{\psi}_{\mathcal{L}}} \,,
\end{align}
which are the {\sl Herglotz--Euler--Lagrange equations}
\eqref{actioneqs} and \eqref{ELeqs2}.
Observe that these last equations can also be obtained from the compatibility equations \eqref{unieqs}, since
 \begin{align}
0 =\, & 
\parderr{L}{x^\mu}{y_\mu^B} 
+\frac{\partial^2L}{\partial y^A \partial y^B_\mu}y_\mu^A 
+\frac{\partial^2L}{\partial y^A_\lambda\partial y^B_\mu}(X_\circ)_{\mu\lambda}^A
+\frac{\partial^2L}{\partial s^\lambda\partial y^B_\mu}(X_\circ)_\mu^\lambda
-\frac{\partial L}{\partial y^B}-\frac{\partial L}{\partial s^\mu}\, p_B^\mu 
\\ = & \parderr{L}{x^\mu}{y_\mu^B}
+\frac{\partial^2L}{\partial y^A \partial y^B_\mu}\derpar{y^A}{x^\mu}
+\frac{\partial^2L}{\partial y^A_\lambda\partial y^B_\mu}\frac{\partial^2y^A}{\partial x^\mu \partial x^\lambda}
+\frac{\partial^2L}{\partial s^\lambda\partial y^B_\mu}\derpar{s^\lambda}{x^\mu}
-\frac{\partial L}{\partial y^B}-\frac{\partial L}{\partial s^\mu}\frac{\partial L}{\partial y^B_\mu} \\
= & \frac{\partial}{\partial x^\mu}
\left(\frac{\displaystyle\partial L}{\partial
y^B_\mu}\circ{\bm{\psi}_\circ}\right)-
\left(\frac{\partial L}{\partial y^B}+
\displaystyle\frac{\partial L}{\partial s^\mu}\displaystyle\frac{\partial L}{\partial y^B_\mu}\right)\circ{\bm{\psi}_\circ} \, .
\end{align}

On the other hand, taking the Hamiltonian function
${\rm H} = p^\mu_Ay^A_\mu-L$, introduced in Remark \ref{remarkHam},
we have that $\displaystyle\derpar{{\rm H}}{p^\mu_A}=y^A_\mu$, $\displaystyle\derpar{{\rm H}}{s^\mu}=-\derpar{L}{s^\mu}$,
$\displaystyle\derpar{{\rm H}}{y^A}=-\derpar{L}{y^A}$, $\displaystyle\derpar{{\rm H}}{x^\mu}=-\derpar{L}{x^\mu}$ and, moreover, 
$\bm{\psi}_{\mathcal{H}}=\overline{\pr}_\circ \circ\bm{\psi}_\circ$
and ${\rm H}=H\circ\overline{\pr}_\circ$.
Then, first,  \eqref{finaleqs1} transforms into
$$
\parder{s^\mu}{x^\mu}=L\circ{\bm{\psi}}_\circ=
\left(p^\mu_Ay^A_\mu-{\rm H}\right)\circ{\bm{\psi}}_\circ=
\bigg(p^\mu_A\,\derpar{{\rm H}}{p^\mu_A}-{\rm H}\bigg)\circ{\bm{\psi}}_\circ=
\bigg(p^\mu_A\,\derpar{H}{p^\mu_A}-H\bigg)\circ{\bm{\psi}}_{\mathcal{H}}\,. 
$$
Second, from equations \eqref{coefs0a} (the semi-holonomy condition), we obtain,
$$
\derpar{y^A}{x^\mu}=y^A_\mu=\derpar{{\rm H}}{p^\mu_A} \circ{\bm{\psi}}_\circ=
\derpar{H}{p^\mu_A}\circ{\bm{\psi}}_{\mathcal{H}}\, .
$$
Finally, from  \eqref{finaleqs2} we get,
$$
\derpar{p^\mu_A}{x^\mu}=
\bigg(\frac{\partial L}{\partial y^A}+\frac{\partial L}{\partial s^\mu} \frac{\partial L}{\partial y^A_\mu}\bigg)\circ{\bm{\psi}}_\circ=
-\bigg(\frac{\partial{\rm H}}{\partial y^A}+\frac{\partial{\rm H}}{\partial s^\mu} p_A^\mu\bigg)\circ{\bm{\psi}}_\circ=
-\bigg(\frac{\partial H}{\partial y^A}+\frac{\partial H}{\partial s^\mu} p_A^\mu\bigg)\circ{\bm{\psi}}_{\mathcal{H}} \, .
$$
In this way we have recovered the {\sl Herglotz--Hamilton--de Donder--Weyl equations} \eqref{HHDWeqs}
\end{remark}

As a final remark, the relation among the Lagrangian, Hamiltonian, and unified Lagrangian--Hamiltonian formalisms can also be established for the field equations for multivector fields \eqref{vfH20} 
(and, as a consequence, for the equivalent field equations for Ehresmann connections \eqref{EcH0}). In fact:

\begin{theorem}
\label{relation2}
Let $({\cal W}_\circ,\Theta_\circ,\omega)$ be a premulticontact Hamiltonian system described by a (hyper)regular Lagrangian.
If ${\bf X}_\circ\in\vf^m(\mathcal{W}_\circ)$ is any integrable $\rho^\circ_M$-transverse multivector field, tangent to $\mathcal{W}_1$, fulfilling s \eqref{vfH20}, 
then the multivector field ${\bf X}_\mathcal{L}=(\pr_\circ)_*{\bf X}_\circ\in\vf^m(\mathcal{P})$
is a holonomic multivector field solution to the Lagrangian field equations \eqref{vfH},
and the multivector field ${\bf X}_\mathcal{H}=(\overline{\pr}_\circ)_*{\bf X}_\circ\in\vf^m(\mathcal{P}^*)$
is a $\overline\tau$-transverse and integrable multivector field solution to the Hamilton--de Donder--Weyl
equations \eqref{vfH2}

Conversely, starting from every holonomic multivector field ${\bf X}_\mathcal{L}\in\vf^m(\mathcal{P})$ and from every $\overline\tau$-transverse and integrable multivector field ${\bf X}_\mathcal{H}\in\vf^m(\mathcal{P}^*)$ 
fulfilling s \eqref{vfH} and \eqref{vfH2}, respectively, we recover 
an integrable $\rho^\circ_M$-transverse multivector field 
${\bf X}_\circ=(\pr_\circ^{-1})_*{\bf X}_\mathcal{L}=(\overline{\pr}_\circ^{-1})_*{\bf X}_\mathcal{H}\in\vf^m(\mathcal{W}_\circ)$
which is a solution to the s \eqref{vfH20}
(remember that, in the (hyper)regular case, ${\pr}_\circ$ and $\overline{\pr}_\circ$ are (local) diffeomorphisms).
\end{theorem}
\begin{proof}
The result is immediate bearing in mind that, in the (hyper)regular case,
${\pr}_\circ$ and $\overline{\pr}_\circ$ are (local) diffeomorphisms
and that, as ${\bf X}_\circ$ is a holonomic vector field, hence so is ${\bf X}_\mathcal{L}$.
\end{proof}

\begin{remark}
As in Remark \ref{singequiv},
in the singular (almost-regular) case, we have an analogous statement,
but taking ${\cal P}_\circ^*$ instead of ${\cal P}_\circ$,
and considering the premulticontact Hamiltonian system
$(\mathcal{P}_\circ^*,\Theta_\mathcal{H}^\circ,\omega)$
(see Diagram \eqref{unified2}).
Then, the proof is the same as for the above theorem, but taking into account that now the projection $\overline{\pr}_1^\circ$ is not a diffeomorphism but a submersion and, 
in particular, the multivector fields ${\bf X}_\circ\in\vf^m(\mathcal{W}_\circ)$ 
must be tangent to $\mathcal{W}_1$ and such that their restrictions to $\mathcal{W}_1$ are $\overline{\pr}_1^\circ$-projectable).
In addition, in the case of singular systems with a
final constraint submanifold ${\cal W}_f\hookrightarrow\mathcal{W}_1$,
the multivector fields must be tangent to this submanifold and
all results must hold at the points of ${\cal W}_f$ instead of ${\cal W}_1$.
\end{remark}

Notice that, in the Hamiltonian formalism, the coefficients $(X_\circ)_{\mu B}^\nu$ of the multivector fields solution to the Hamilton--de Donder--Weyl equations \eqref{vfH2}
are not related to the coefficients $(X_\circ)_{\mu\lambda}^A$ of the multivector fields solution to the Lagrangian equations \eqref{vfH}
(as in the unified formalism, equation \eqref{tangeqs}) 
as long as the equivalence between both Lagrangian and Hamiltonian formalisms is not established through the Legendre map.

\section{Maxwell's electromagnetism with action dependent terms}
\label{example}

As an example, we study the Skinner--Rusk multicontact formalism of Maxwell's equations with charges and currents, and a non-conservative term.

The Lagrangian and/or Hamiltonian formalism for 
a nonconservative version of Maxwell's equations was first studied in \cite{LPAF_2018},
in the context of contact geometry. 
Later, it was formalized in \cite{GM_22,GRR_22}, using $k$-contact structures,
and in \cite{LGMRR_23,LGMRR_25}, using multicontact geometry.

Let $M$ be a $4$-dimensional manifold representing spacetime, $P \rightarrow M$ the principal bundle with structure group $U(1)$, $\pi\colon C\rightarrow M$ the associated bundle of connections \cite{CM_01},
and $\bar\pi^1\colon J^1\pi\to M$ the corresponding jet bundle.
In order to state the unified Skinner--Rusk formalism of this theory,
consider the extended jet-multimomentum bundle
$\mathcal{W}=(J^1 \pi \times_E \mathcal{M}\pi) \times_M \bigwedge\nolimits^{m-1} \mathrm{T}^*M$, 
with local coordinates $(x^\mu,A_\mu,A_{\mu,\nu},p^{\mu,\nu},p,s^\mu)$ such that the volume form is expressed as $\omega=\d x^0\wedge\d x^1\wedge\d x^2\wedge\d x^3$,
and where $A_\mu$ are the components of the electromagnetic four-potential
and $A_{\mu,\nu}$ and $p^{\mu,\nu}$ denote the corresponding multivelocities and multimomenta.

The electromagnetic Lagrangian with a linear dissipation term is taken to be
\[
L=-\frac{1}{4\mu_0}g^{\alpha\mu}g^{\beta\nu}F_{\mu\nu}F_{\alpha\beta}-A_\alpha J^\alpha-\gamma_\alpha s^\alpha\,,
\]
where $F_{\mu\nu}=A_{\nu,\mu}-A_{\mu,\nu}$ is the electromagnetic tensor field, $J^\alpha$ (the electromagnetic four-current) and $\gamma_\alpha$ are smooth functions on $M$ (for $0\leq\alpha\leq 3$), $g^{\mu\nu}$ is a Lorentzian metric on $M$, 
and $\mu_0$ is a constant.

The constraint function that defines the Hamiltonian submanifold $\mathcal{W}_\circ\hookrightarrow\mathcal{W}$ is
\[
C-L=p+p^{\mu,\nu}A_{\mu,\nu}+\frac{1}{4\mu_0}g^{\alpha\mu}g^{\beta\nu}F_{\mu\nu}F_{\alpha\beta}+A_\alpha J^\alpha+\gamma_\alpha s^\alpha= 0 \, ,
\]
and we have the $m$-form
\[
\Theta_\circ = -p^{\alpha, \mu} \, \d A_\alpha \wedge \d ^3x_\mu - (L- A_{\mu, \alpha} \, p^{\alpha, \mu}) \, \d^4x + \d s^\mu\wedge\d ^3x_\mu\in\df^m(\mathcal{W}_\circ) \,.
\]
Then, $(\mathcal{W}_\circ,\Theta_\circ,\omega)$ is a special premulticontact bundle, on the points of the submanifold
\begin{align*}
\mathcal{W}_1:= &\{ w\in\mathcal{W}_\circ\mid
(\inn{X}\d\Theta_\circ)(w)=0\,,\quad \text{ for every } X\in\vf^{{\rm V}(\overline{\pr})}(\mathcal{W}_\circ)\} \\
= & \left\{ (x^\mu,A_\mu,A_{\mu,\alpha},p^{\mu,\alpha},s^\mu)\in \mathcal{W}_\circ \mid  p^{\mu,\alpha}=\derpar{L}{A_{\mu,\alpha}} =\frac{1}{\mu_0}g^{\mu\nu}g^{\alpha\beta}F_{\beta\nu} \right\} \, .
\end{align*}
Indeed,
$$
   \mathcal{C} \underset{\mathcal{W}_1}{=} \left<\frac{\partial}{\partial A_{\mu,\nu}}\right>_{\mu,\nu=0,1,2,3}
   \,\quad \text{ and } \qquad
   \mathcal{D}^\mathfrak{R}\underset{\mathcal{W}_1}{=}\left<\frac{\partial}{\partial A_{\mu,\nu}},\frac{\partial}{\partial s^\mu}\right>_{\mu,\nu=0,1,2,3}\,. 
$$
and the last condition of Definition \ref{multicontactdef} is satisfied because $\displaystyle \innp{\frac{\partial}{\partial s^\mu}}\Theta_\circ=\d^3x_\mu$.
The dissipation form is $\sigma_{\Theta_\circ} = \gamma_\mu \d x^\mu$, which is closed as long as $\displaystyle\frac{\partial \gamma_\mu}{\partial x^\nu}=\frac{\partial \gamma_\nu}{\partial x^\mu}$.
Thus, we have
\begin{equation}
\begin{aligned}
\bd\Theta_\circ =& -\d p^{\alpha, \mu} \wedge \d A_\alpha \wedge \d ^3x_\mu + (J^\alpha + \gamma_\mu p^{\alpha, \mu})\, \d A_\alpha \wedge \d^4 x \\ - & \left(\frac{1}{\mu_0} g^{\alpha \beta} g^{\mu \nu} F_{\beta \nu}- p^{\alpha, \mu}\right) \d A_{\mu, \alpha} \wedge \d^4x + A_{\mu, \alpha}\d p^{\alpha, \mu}\wedge\d ^4x \\
\underset{\mathcal{W}_1}{=} & -\d p^{\alpha, \mu} \wedge \d A_\alpha \wedge \d ^3x_\mu + \left(J^\alpha + \frac{\gamma_\mu}{\mu_0}g^{\mu\nu}g^{\alpha\beta}F_{\beta\nu}\right) \d A_\alpha \wedge \d^4 x + A_{\mu, \alpha}\d p^{\alpha, \mu}\wedge\d ^4x \, .
\end{aligned}
\end{equation}

The expression in a natural chart of coordinates of ${\cal W}_\circ$ of a $\rho^\circ_M$-transverse and locally decomposable $m$-multivector field
satisfying $\inn{\bfX_\circ}\omega=1$ is
$$
{\bf X}_\circ= \bigwedge_{\rho=1}^m
\bigg(\parder{}{x^\rho}+(X_\circ)_{\rho,\alpha}\derpar{}{A_\alpha}+(X_\circ)_{\rho\mu,\alpha}\derpar{}{A_{\mu,\alpha}}+
(X_\circ)^{\mu,\alpha}_\rho\derpar{}{p^{\mu,\alpha}}+(X_\circ)_\rho^\mu\,\derpar{}{s^\mu}\bigg) \,.
$$
For such a multivector field, the field equations \eqref{vfH20} lead to
\begin{equation}
    \begin{dcases}
    \, (X_\circ)^\mu_\mu=-\frac{1}{4\mu_0}g^{\alpha\mu}g^{\beta\nu}F_{\mu\nu}F_{\alpha\beta}-A_\alpha J^\alpha-\gamma_\alpha s^\alpha \,,\\
    \, (X_\circ)_{\mu,\alpha} = A_{\mu,\alpha} \,, \\
    \, (X_\circ)^{\mu,\alpha}_\mu=-\left(J^\alpha +\gamma_\mu \, p^{\mu, \alpha}\right) \,, \\
    \, p^{\mu,\alpha}=\frac{1}{\mu_0}g^{\mu\nu}g^{\alpha\beta}F_{\beta\nu} \, .
    \label{EM1}
    \end{dcases}
\end{equation}

Thus, as expected in the Skinner--Rusk formalism, we recover the holonomy condition and the Legendre map. 
For the integral sections $\bm\psi_\circ(x^\nu)=\big( x^\nu,A_\mu(x^\nu),A_{\mu,\alpha}(x^\nu),p^{\mu,\alpha}(x^\nu),s^\mu(x^\nu) \big)$ of the integrable multivector fields ${\bf X}_\circ$
solutions to these equations, the system \eqref{EM1} leads to
\begin{equation}
    \begin{dcases}
    \, \derpar{s^\mu}{x^\mu}=-\frac{1}{4\mu_0}g^{\alpha\mu}g^{\beta\nu}F_{\mu\nu}F_{\alpha\beta}-A_\alpha J^\alpha-\gamma_\alpha s^\alpha \,,\\
    \, \derpar{A_\mu}{x^\alpha}=A_{\mu,\alpha} \,, \\
    \, \derpar{p^{\mu,\alpha}}{x^\mu}=-\left(J^\alpha +\gamma_\mu \, p^{\mu, \alpha}\right) \,, \\
    \, p^{\mu,\alpha}=\frac{1}{\mu_0}g^{\mu\nu}g^{\alpha\beta}F_{\beta\nu}\ \,.
    \label{EM2}
    \end{dcases}
\end{equation}
which are also the coordinate expression of s \eqref{Hamilton-Lagrange eq} for this theory.
Now, following the procedure described in Remark \ref{recovering},
first we recover the Herglotz--Euler--Lagrange equations,
\begin{equation}
    \begin{dcases}
    \frac{\partial s^\mu}{\partial x^\mu}=-\frac{1}{4\mu_0}g^{\alpha\mu}g^{\beta\nu}F_{\mu\nu}F_{\alpha\beta}-A_\alpha J^\alpha-\gamma_\alpha s^\alpha   \, , \\
    \mu_0 J^\mu = - g^{\nu\alpha}g^{\mu\sigma}\Big(\frac{\partial F_{\sigma\alpha}}{\partial x^\nu}+\gamma_\nu F_{\sigma\alpha}\Big) \,, 
    \label{EM3}
    \end{dcases}
\end{equation}
and second the Herglotz--Hamilton--de Donder--Weyl equations,
\begin{equation}
    \begin{dcases}
    \frac{\partial s^\mu}{\partial x^\mu}=-\frac{\mu_0}{4}g_{\alpha\mu}g_{\beta\nu}p^{\mu\nu}p^{\alpha\beta}-A_\alpha J^\alpha-\gamma_\alpha s^\alpha \, , \\
    \derpar{A_\mu}{x^\alpha}=\frac{\mu_0}{2}\,g_{\alpha\nu}g_{\mu\beta}\big(p^{\beta,\nu}-p^{\nu,\beta} \big) \,, \\
   J^\mu=-\derpar{p^{\nu,\mu}}{x^\nu}-\gamma_\nu p^{\nu,\mu} \,.
    \end{dcases}
\end{equation}
The physical meaning of these equations and their interpretation as an electromagnetic field in matter is discussed in \cite{LGMRR_23,LGMRR_25,GM_22,LPAF_2018}.

\section{Conclusions and outlook}

The geometric covariant description of conservative classical field theories is well-established and primarily relies on multisymplectic geometry. In contrast, describing action-dependent (non-conservative) field theories requires an analogous framework, the so-called multicontact structures, which generalize both contact and multisymplectic geometries.

This work first has presented a comprehensive review of the theory of multicontact structures, and how this geometric setting allows the geometric description of the Lagrangian and Hamiltonian formalisms action-dependent field theories.
The paper's main original contribution has been the development of a unified formulation that integrates these two formalisms,
which is presented as a generalization of the Skinner--Rusk formalism.

Thus, the {\sl extended} and {\sl restricted jet-multimomentum} (or {\sl Pontryagin}) {\sl bundles},
$\mathcal{W}$ and $\mathcal{W}_r$ are introduced.
The first one, $\mathcal{W}$, carries a canonical premulticontact structure, but the formalism takes place in
the {\sl Hamiltonian submanifold} $\mathcal{W}_\circ$ of $\mathcal{W}$, which is diffeomorphic to $\mathcal{W}_r$.
This manifold $\mathcal{W}_\circ$ is endowed with a premulticontact structure, which is induced by that on $\mathcal{W}$,
but not with a special (pre)multicontact structure, as required to model an action-dependent field theory.
This is achieved by defining a new submanifold $\mathcal{W}_1$ in $\mathcal{W}_\circ$ at which points
the above premulticontact structure becomes a special premulticontact one.
The resulting premulticontact field equations contain all the essential information about the Lagrangian and Hamiltonian descriptions of the theory, namely:
\begin{enumerate}
\item 
The algebraic equations \eqref{coefs0c},
which are the compatibility constraints for the field equations,
and hold identically, since they are also the constraints \eqref{contraints1}
giving the submanifold $\mathcal{W}_1$ in $\mathcal{W}_\circ$,
where the special premulticontact structure of $\mathcal{W}_\circ$ is defined.
This double character of these constraints is a feature of the unified formalism inherent in action-dependent field theories.
As is usual in this formalism,
these constraints also give the Legendre map.
\item 
The semi-holonomy conditions given by the equations \eqref{coefs0a}, which force the sections solutions to the field equations to be holonomic.
These conditions are achieved independently of the regularity of the Lagrangian.
\item 
The Herglotz--Euler--Lagrange equations \eqref{ELeqs2}.
\item 
The Herglotz--Hamilton--de Donder--Weyl equations \eqref{HHDWeqs}.
\item 
The Lagrangian and Hamiltonian variational equations \eqref{actioneqs} and \eqref{actionHameqs}.
\end{enumerate}
In addition, the relation between the unified Skinner-Rusk formalism and the Lagrangian and Hamiltonian formalisms is proved both for (hyper) regular and singular (almost-regular) cases.

This formalism opens possibilities for future extension to higher-order action-dependent field theories and its application to modified action-dependent theories of gravitation;
as well as other interesting cases, such as action-dependent modified Kaluza--Klein and Yang--Mills theories.

\subsection*{Acknowledgements}

We acknowledge the financial support of the 
{\sl Ministerio de Ciencia, Innovaci\'on y Universidades} (Spain), projects PID2021-125515NB-C21 (MCIN/AEI/10.13039/501100011033/FEDER,UE), and RED2022-134301-T of AEI,
and Ministry of Research and Universities of
the Catalan Government, project 2021 SGR 00603 \textsl{Geometry of Manifolds and Applications, GEOMVAP}.


\bibliographystyle{abbrv}
{\small
\bibliography{references.bib}
}

\end{document}